\documentclass[a4paper,UKenglish,cleveref, autoref, thm-restate]{lipics-v2021}

\usepackage{amsmath, amssymb,amsthm}

\usepackage{graphicx}
\usepackage{caption}
\usepackage{subcaption}
\usepackage{longtable}
\usepackage{booktabs}
\usepackage{tabu}
\usepackage[referable]{threeparttablex}
\usepackage{multirow}
\usepackage{subcaption}
\usepackage{hyperref}
\usepackage[capitalise]{cleveref}

\crefname{figure}{Figure}{Figure}
\crefformat{footnote}{#2\footnotemark[#1]#3}
\usepackage{tikz}
\usepackage{comment}
\usepackage{parskip}
\usepackage{paralist}
\usepackage{pgfplots}

\usepackage[ruled, linesnumbered, noend]{algorithm2e}
\usepackage{tikz}
\usetikzlibrary{arrows,automata,shapes,decorations,decorations.markings,calc, matrix,decorations.pathmorphing,snakes}

\usepackage{xspace}
\usepackage{xcolor}
\usepackage[inline, shortlabels]{enumitem}

\usepackage{multicol}
\usepackage{bbm}

\usepackage{thm-restate}

\tikzstyle{vertex} = [rectangle]
\tikzstyle{arrow} = [thick,->,>=latex]
\tikzstyle{arrow2} = [thin,->,>=latex, draw=gray]
\tikzstyle{textbox} = [rectangle, text centered, text width=6cm, draw=black]

\newcommand{\figlabel}[1]{\label{fig:#1}}
\newcommand{\figref}[1]{Figure~\ref{fig:#1}}
\newcommand{\seclabel}[1]{\label{sec:#1}}
\newcommand{\secref}[1]{Section~\ref{sec:#1}}

\newcommand{\deflabel}[1]{\label{def:#1}}

\newcommand{\problabel}[1]{\label{prob:#1}}

\newcommand{\thmlabel}[1]{\label{thm:#1}}
\newcommand{\thmref}[1]{Theorem~\ref{thm:#1}}

\newcommand{\lemlabel}[1]{\label{lem:#1}}
\newcommand{\lemref}[1]{Lemma~\ref{lem:#1}}

\newcommand{\applabel}[1]{\label{app:#1}}
\newcommand{\appref}[1]{Appendix~\ref{app:#1}}
\newcommand{\algolabel}[1]{\label{algo:#1}}
\newcommand{\algoref}[1]{Algorithm~\ref{algo:#1}}

\newtheorem{problem}{Problem}
\newtheorem{hypothesis}{Hypothesis}

\newenvironment{myparagraph}[1]{\noindent{\bf #1.}}{}

 \newcommand{\Paragraph}[1]{\noindent{\bf #1}}
 \newcommand{\SubParagraph}[1]{\noindent{\em #1}}

\DeclareMathAlphabet{\mathpzc}{OT1}{pzc}{m}{it}

\mathchardef\mhyphen="2D 

\newcommand{\Natsinf}{\mathbb{N}\cup\{\infty\}}

\newcommand{\ov}{\overline}

\newcommand{\set}[1]{\{#1\}}
\newcommand{\setpred}[2]{\{#1 \,|\, #2\}}
\newcommand{\proj}[2]{#1|_{#2}}
\renewcommand{\emptyset}{\varnothing}

\newcommand{\tr}{\sigma}
\newcommand{\lk}{\ell}

\newcommand{\lab}{\mathsf{lab}}
\newcommand{\ev}[1]{\langle #1\rangle}
\newcommand{\events}[1]{\mathsf{Events}_{#1}}
\newcommand{\threads}[1]{\mathsf{Threads}_{#1}}
\newcommand{\locks}[1]{\mathsf{Locks}_{#1}}
\newcommand{\vars}[1]{\mathsf{Vars}_{#1}}
\newcommand{\reads}[1]{\mathsf{Reads}_{#1}}
\newcommand{\writes}[1]{\mathsf{Writes}_{#1}}
\newcommand{\acquires}[1]{\mathsf{Acquires}_{#1}}
\newcommand{\releases}[1]{\mathsf{Releases}_{#1}}
\newcommand{\accesses}[1]{\mathsf{Accesses}_{#1}}

\newcommand{\hbgraph}{\textsf{G}(\hb{\tr})}

\newcommand{\PTime}{\textsf{P}}

\newcommand{\NPhard}{\textsf{NP-hard}}
\newcommand{\SAT}{\textsf{SAT}}
\newcommand{\poly}{\operatorname{poly}}

\newcommand{\opfont}[1]{\texttt{#1}}

\newcommand{\critsec}{\opfont{cs}}
\newcommand{\acq}{\opfont{acq}}
\newcommand{\rel}{\opfont{rel}}
\newcommand{\rd}{\opfont{r}}
\newcommand{\wt}{\opfont{w}}

\newcommand{\sync}{\opfont{sync}}

\newcommand{\Thread}{t}
\newcommand{\ThreadOf}[1]{\tid(#1)}
\newcommand{\OpOf}[1]{\Operation(#1)}

\newcommand{\match}[1]{\mathsf{match}_{#1}}
\newcommand{\crit}[1]{\mathsf{CS}_{#1}}
\newcommand{\lheld}[1]{\mathsf{locksHeld}_{#1}}

\newcommand{\orthv}{\textsf{OV}}
\newcommand{\orthvk}[1]{\textsf{#1\--\orthv}}
\newcommand{\hs}{\textsf{HS}}
\newcommand{\multiconn}{\textsf{MCONN}}
\newcommand{\fofee}{\textsf{FO}($\mathsf{\forall\exists\exists}$)}
\newcommand{\foeee}{\textsf{FO}($\mathsf{\exists\exists\exists}$)}
\newcommand{\lw}[1]{\mathsf{lw}_{#1}}

\newcommand{\ntime}[1]{\textsf{NTIME}[\mathsf{#1}]}
\newcommand{\contime}[1]{\textsf{coNTIME}[\mathsf{#1}]}

\newcommand{\seth}{\textsf{SETH}}
\newcommand{\nseth}{\textsf{NSETH}}
\newcommand{\ovc}{\textsf{OVH}}
\newcommand{\hsc}{\textsf{HSH}}
\newcommand{\foh}{\textsf{FOPH}($\mathsf{\forall\exists\exists}$)}

\newcommand{\ovinstwo}{\orthv(n,d)}
\newcommand{\ovinsthree}[1]{\orthv(n, d, #1)}
\newcommand{\hsins}{\textsf{HS}(n,d)}

\newcommand{\fgr}[2]{${}_{(\mathsf{#1})}\!\!\preceq_{(\mathsf{#2})}$}

\newcommand{\acr}[1]{\textsf{#1}}
\newcommand{\acrtr}{\operatorname{\acr{tr}}}
\newcommand{\TO}{\operatorname{\acr{TO}}}

\newcommand{\HB}{\operatorname{\acr{HB}}}

\newcommand{\ord}[2]{\leq^{#1}_{\mathsf{#2}}}

\newcommand{\strictord}[2]{<^{#1}_{\mathsf{#2}}}

\newcommand{\trord}[1]{\ord{#1}{\acrtr}}
\newcommand{\stricttrord}[1]{\strictord{#1}{\acrtr}}
\newcommand{\tho}[1]{\ord{#1}{\TO}}

\newcommand{\hb}[1]{\ord{#1}{\HB}}

\newcommand{\Path}{\rightsquigarrow}

\newcommand{\tid}{\operatorname{\mathsf{tid}}}

\newcommand{\mx}{\sqcup}
\newcommand{\mn}{\sqcap}
\newcommand{\acqls}[1]{\mathsf{AcqLS}^{#1}}
\newcommand{\rells}[1]{\mathsf{RelLS}^{#1}}
\newcommand{\idxlk}{\mathsf{p}}
\newcommand{\setrdls}{\mathsf{S}}

\SetInd{0.5em}{0.5em}

\SetKw{Break}{break}
\SetKwProg{MyFunction}{function}{}{}
\SetKwProg{MyRoutine}{routine}{}{}
\SetKwProg{MyProcedure}{procedure}{}{}
\SetKwFunction{FunctionCTInitialize}{{Init}}
\SetKwFunction{FunctionCTLessThan}{LessThan}
\SetKwFunction{FunctionCTJoin}{Join}
\SetKwFunction{FunctionVCJoin}{Join}
\SetKwFunction{FunctionCTSubRootJoin}{SubRootJoin}
\SetKwFunction{FunctionVCCopy}{Copy}
\SetKwFunction{FunctionCTMonotoneCopy}{MonotoneCopy}
\SetKwFunction{FunctionCTCopyCheckMonotone}{CopyCheckMonotone}
\SetKwFunction{FunctionCTVTime}{VectorTime}
\SetKwFunction{FunctionCTIncrement}{Increment}
\SetKwFunction{FunctionVCIncrement}{Increment}
\SetKwFunction{FunctionCTGet}{Get}
\SetKwFunction{RoutineTestDeque}{TestDeque}
\SetKwFunction{RoutineUpdatedNodesForJoin}{getUpdatedNodesForJoin}
\SetKwFunction{RoutineUpdatedNodesForCopy}{getUpdatedNodesForCopy}
\SetKwFunction{RoutineDetachNodes}{detachNodes}
\SetKwFunction{RoutineAttachNodesJoin}{attachNodesJoin}
\SetKwFunction{RoutineAttachNodesCopy}{attachNodesCopy}
\SetKwFunction{RoutineAttachNodes}{attachNodes}
\SetKwFunction{RoutinePushChild}{pushChild}
\SetKwFunction{RoutineInsertChild}{insertChild}
\SetKwFunction{ProcInitialize}{initialize}
\SetKwFunction{ProcAcquire}{acquire}
\SetKwFunction{ProcRelease}{release}
\SetKwFunction{ProcWrite}{write}
\SetKwFunction{ProcRead}{read}

\newcommand{\CTC}{\mathbb{C}}

\newcommand{\CTL}{\mathbb{L}}

\newcommand{\VCW}{\mathbb{W}}

\newcommand{\Andreas}[1]{\textcolor{blue}{\textbf{Andreas:}~#1}}

\newcommand{\Operation}{\operatorname{op}}
\newcommand{\NumEvents}{\mathcal{N}}
\newcommand{\NumThreads}{\mathcal{T}}
\newcommand{\NumVariables}{\mathcal{V}}
\newcommand{\NumLocks}{\mathcal{L}}

\newcommand{\eraser}{\textsc{Eraser}\xspace}

\newcommand{\djit}{\textsc{Djit}\xspace}

\newcommand{\goldilocks}{\textsf{GoldiLocks}\xspace}
\newcommand{\fasttrack}{\textsc{FastTrack}\xspace}

\usetikzlibrary{arrows.meta,calc,decorations.markings,math,arrows.meta}

\pgfdeclarelayer{bg}    
\pgfsetlayers{bg,main}  

\SetKw{Continue}{continue}

\SetCommentSty{mycommfont}

\makeatletter
\newcommand*{\centerfloat}{%
  \parindent \z@
  \leftskip \z@ \@plus 1fil \@minus \textwidth
  \rightskip\leftskip
  \parfillskip \z@skip}
\makeatother

\newcommand{\execution}[2]{
\scalebox{1}{
  \begin{tikzpicture}%
    \foreach \x in {1,...,#1}
    \node[right] at (1.5*\x+0.2,0.25) {$\Thread_{\x}$};
    \draw (1.2,0) -- (#1*1.5+1.3,0);%
    \pgfmathsetmacro{\y}{1};%
    #2%
    \draw (1.2,0) -- (1.2,-0.4*\y);%
    \draw (#1*1.5+1.3,0) -- (#1*1.5+1.3,-0.4*\y);%
    \foreach \x in {2,...,#1}
    \draw (1.2,-0.4*\y) -- (#1*1.5+1.3,-0.4*\y);%
  \end{tikzpicture}
}
}

\newcommand{\executionOffset}[3]{
\scalebox{1}{
  \begin{tikzpicture}%
    \foreach \x in {1}{
    \node[right] at (1.5*\x+0.2,0.25)    {$\Thread_{#2}$};
    };
    \draw (1.2,0) -- (#1*1.5+1.3,0);%
    \pgfmathsetmacro{\y}{1};%
    #3%
    \draw (1.2,0) -- (1.2,-0.4*\y);%
    \draw (#1*1.5+1.3,0) -- (#1*1.5+1.3,-0.4*\y);%
    \foreach \x in {2,...,#1}
    \draw (1.2,-0.4*\y) -- (#1*1.5+1.3,-0.4*\y);%
  \end{tikzpicture}
}
}

\newcommand{\leanExecIndex}[2]{
\scalebox{1}{
  \begin{tikzpicture}%
    \node[right] at (1.7,0.25)    {$\Thread_{#1}$};
    \draw (1.2,0) -- (2.8,0);%
    \pgfmathsetmacro{\y}{1};%
    #2%
  \end{tikzpicture}
}
}

\newcommand{\figev}[2]{
\pgfmathsetmacro{\y}{\y+1};
\pgfmathsetmacro{\y}{\y-1};
\node [left] at (1.25,-0.4*\y)  {\pgfmathprintnumber{\y}};%
\node at (#1*1.5 + 0.45,-0.4*\y) { #2 };%
\pgfmathsetmacro{\y}{\y+1};
}

\newcommand{\Leanfigev}[2]{
\pgfmathsetmacro{\y}{\y+1};
\pgfmathsetmacro{\y}{\y-1};
\node at (#1*1.5 + 0.45,-0.4*\y) { #2 };%
\pgfmathsetmacro{\y}{\y+1};
}

\newcommand{\figevStale}[2]{
\pgfmathsetmacro{\y}{\y+1};
\pgfmathsetmacro{\y}{\y-1};
\node at (#1*1.5 + 0.45,-0.4*\y) { #2 };%
\pgfmathsetmacro{\y}{\y+1};
}

\newcommand{\figevoffset}[3]{
\pgfmathparse{#1}
\pgfmathsetmacro{\offset}{\pgfmathresult};
\pgfmathparse{\y+\offset}
\pgfmathsetmacro{\newindex}{\pgfmathresult};
\node [left] at (1.25,-0.4*\y)  {\pgfmathprintnumber{\newindex}};%
\node at (#2*1.5 + 0.45,-0.4*\y) {#3};%
\pgfmathsetmacro{\y}{\y+1};
}

\newcommand{\executionfull}[9]{
\scalebox{#3}{
  \begin{tikzpicture}
    \foreach \x in {1,...,#1}
    \node[right] at (#4*\x+#6, #8) {$t_{\x}$}; 
    \draw (#7,0) -- (#1*#4+#7,0); 
    \pgfmathsetmacro{\y}{1};
    #2 
    \draw (#7,0) -- (#7,-#5*\y); 
    \draw (#1*#4+#7,0) -- (#1*#4+#7,-#5*\y); 
    \draw (#7,-#5*\y) -- (#1*#4+#7,-#5*\y); 
    \ifthenelse{#9 = 1}{
      \foreach \x in {2,...,#1}
      \draw[dashed] (#4*\x+#7-#4,0) -- (#4*\x+#7-#4,-#5*\y); 
    }{}
  \end{tikzpicture}
}
}

\newcommand{\figevfull}[9]{
\ifthenelse{#7 = 1}{
  \ifthenelse{#8 = -1}{
    \node [left] at (#5,(-#4*\y)  {\pgfmathprintnumber{\y}};%
  }{
    \node [left] at (#5,(-#4*\y)  {#9};%
  }
}{}
\node at (#1*#3 + #6,(-#4*\y) {$ #2 $};%
\pgfmathsetmacro{\y}{\y+1};
}

\newcommand{\HS}[2]{
\scalebox{1}{
  \begin{tikzpicture}%
    \node[right] at (1.7,0.25)  {$X$};
    \draw (1.2,0) -- (2.8,0);%
    \pgfmathsetmacro{\y}{1};%
    #1%
    \draw (1.2,0) -- (1.2,-0.4*\y);%
    \draw (2.8,0) -- (2.8,-0.4*\y);%
    \draw (1.2,-0.4*\y) -- (2.8,-0.4*\y);%
    
    \node[right] at (3.5,0.25)  {$Y$};
    \draw (3,0) -- (4.6,0);%
    \pgfmathsetmacro{\y}{1};%
    #2%
    \draw (3,0) -- (3,-0.4*\y);%
    \draw (4.6,0) -- (4.6,-0.4*\y);%
    \draw (3,-0.4*\y) -- (4.6,-0.4*\y);%
  \end{tikzpicture}
}
}

\newcommand{\hsevX}[1]{
\pgfmathsetmacro{\y}{\y+1};
\pgfmathsetmacro{\y}{\y-1};
\node [left] at (1.25,-0.4*\y)  {\pgfmathprintnumber{\y}};%
\node at (1.5 + 0.45,-0.4*\y) { #1 };%
\pgfmathsetmacro{\y}{\y+1};
}

\newcommand{\hsevY}[1]{
\pgfmathsetmacro{\y}{\y+1};
\pgfmathsetmacro{\y}{\y-1};
\node at (3.3 + 0.45,-0.4*\y) { #1 };%
\pgfmathsetmacro{\y}{\y+1};
}

\newcommand{\OVTwo}[2]{
\scalebox{1}{
  \begin{tikzpicture}%
    \node[right] at (1.3,0.25) {$A_1$};
    \draw (1.2,0) -- (2,0);%
    \pgfmathsetmacro{\y}{1};%
    #1%
    \draw (1.2,0) -- (1.2,-0.4*\y);%
    \draw (2,0) -- (2,-0.4*\y);%
    \draw (1.2,-0.4*\y) -- (2,-0.4*\y);%
    
    \node[right] at (2.2,0.25)  {$A_2$};
    \draw (2.1,0) -- (2.9,0);%
    \pgfmathsetmacro{\y}{1};%
    #2%
    \draw (2.1,0) -- (2.1,-0.4*\y);%
    \draw (2.9,0) -- (2.9,-0.4*\y);%
    \draw (2.1,-0.4*\y) -- (2.9,-0.4*\y);%
  \end{tikzpicture}
}
}

\newcommand{\OV}[3]{
\scalebox{1}{
  \begin{tikzpicture}%
    \node[right] at (1.3,0.25) {$A_1$};
    \draw (1.2,0) -- (2,0);%
    \pgfmathsetmacro{\y}{1};%
    #1%
    \draw (1.2,0) -- (1.2,-0.4*\y);%
    \draw (2,0) -- (2,-0.4*\y);%
    \draw (1.2,-0.4*\y) -- (2,-0.4*\y);%
    
    \node[right] at (2.2,0.25)  {$A_2$};
    \draw (2.1,0) -- (2.9,0);%
    \pgfmathsetmacro{\y}{1};%
    #2%
    \draw (2.1,0) -- (2.1,-0.4*\y);%
    \draw (2.9,0) -- (2.9,-0.4*\y);%
    \draw (2.1,-0.4*\y) -- (2.9,-0.4*\y);%
    
    \node[right] at (3.1,0.25)  {$A_3$};
    \draw (3,0) -- (3.9,0);%
    \pgfmathsetmacro{\y}{1};%
    #3%
    \draw (3,0) -- (3,-0.4*\y);%
    \draw (3.9,0) -- (3.9,-0.4*\y);%
    \draw (3,-0.4*\y) -- (3.9,-0.4*\y);%
  \end{tikzpicture}
}
}

\newcommand{\ovX}[1]{
\pgfmathsetmacro{\y}{\y+1};
\pgfmathsetmacro{\y}{\y-1};
\node at (1.2 + 0.4,-0.4*\y) { #1 };%
\pgfmathsetmacro{\y}{\y+1};
}

\newcommand{\ovY}[1]{
\pgfmathsetmacro{\y}{\y+1};
\pgfmathsetmacro{\y}{\y-1};
\node at (2.1 + 0.45,-0.4*\y) { #1 };%
\pgfmathsetmacro{\y}{\y+1};
}

\newcommand{\ovZ}[1]{
\pgfmathsetmacro{\y}{\y+1};
\pgfmathsetmacro{\y}{\y-1};
\node at (3 + 0.45,-0.4*\y) { #1 };%
\pgfmathsetmacro{\y}{\y+1};
}

\SetKwProg{myfun}{function}{}{}
\SetKwProg{myhandler}{}{:}{}
\SetKwFunction{init}{Initialization}
\SetKwFunction{fixpoint}{ComputeSRFClosure}
\SetKwFunction{maxLB}{maxLowerBound}
\SetKwFunction{checkRace}{checkRace}
\SetKwFunction{rdhandler}{read}
\SetKwFunction{wthandler}{write}
\SetKwFunction{acqhandler}{acquire}
\SetKwFunction{relhandler}{release}
\SetKwInput{Input}{Input}
\SetKwInOut{Output}{Output}
\SetKw{Let}{let}
\SetKw{Break}{break}
\SetKw{NOT}{not}
\SetKw{declare}{declare}
\SetKw{exit}{exit}
\SetKwFor{Foreach}{for each}{}{}%
\DontPrintSemicolon

\newcommand{\cle}{\sqsubseteq}

\makeatletter
\newcommand*\bigcdot{\mathpalette\bigcdot@{.5}}
\newcommand*\bigcdot@[2]{\mathbin{\vcenter{\hbox{\scalebox{#2}{$\m@th#1\bullet$}}}}}
\makeatother

\title{Dynamic Data-Race Detection through the Fine-Grained Lens}

\author{Rucha Kulkarni}{University of Illinois at Urbana-Champaign, USA}{ruchark2@illinois.edu}{}{}

\author{Umang Mathur}{University of Illinois at Urbana-Champaign, USA}{umathur3@illinois.edu}{https://orcid.org/0000-0002-7610-0660}{Umang Mathur was partially funded by a Google PhD Fellowship and by the Simons Institute for the Theory of Computing}

\author{Andreas Pavlogiannis}{Aarhus University, Denmark}{pavlogiannis@cs.au.dk}{}{}

\authorrunning{R. Kulkarni and U. Mathur and A. Pavlogiannis} 

\Copyright{Rucha Kulkarni and Umang Mathur and Andreas Pavlogiannis} 

\ccsdesc[500]{Software and its engineering~Software testing and debugging}
\ccsdesc[500]{Theory of computation~Parameterized complexity and exact algorithms}

\keywords{dynamic analyses, data races, fine-grained complexity} 

\category{} 

\relatedversion{} 

\nolinenumbers 

\EventEditors{John Q. Open and Joan R. Access}
\EventNoEds{2}
\EventLongTitle{42nd Conference on Very Important Topics (CVIT 2016)}
\EventShortTitle{CVIT 2016}
\EventAcronym{CVIT}
\EventYear{2016}
\EventDate{December 24--27, 2016}
\EventLocation{Little Whinging, United Kingdom}
\EventLogo{}
\SeriesVolume{42}
\ArticleNo{23}

\renewcommand{\smallskip}{}
\begin{document}

\maketitle


\begin{abstract}
Data races are among the most common bugs in concurrency.
The standard approach to data-race detection is via dynamic analyses, which 
work over executions of concurrent programs, instead of the program source code.
The rich literature on the topic has created various notions of dynamic data races, which are known to be detected efficiently when certain parameters (e.g., number of threads) are small. 
However, the \emph{fine-grained} complexity of all these notions of races has remained elusive, making it impossible to characterize their trade-offs between precision and efficiency.

In this work we establish several fine-grained separations between many popular notions of dynamic data races.
The input is an execution trace $\tr$ with $\NumEvents$ events, $\NumThreads$ threads and $\NumLocks$ locks.
Our main results are as follows.
First, we show that \emph{happens-before (HB) races} can be detected in $O(\NumEvents\cdot \min(\NumThreads, \NumLocks))$ time, improving over the standard $O(\NumEvents\cdot \NumThreads)$ bound when $\NumLocks=o(\NumThreads)$.
Moreover, we show that even reporting an $\HB$ race that involves a read access is hard for 2-orthogonal vectors (2-OV).
This is the first rigorous proof of the conjectured quadratic lower-bound in detecting HB races.
Second, we show that the recently introduced \emph{synchronization-preserving races} are hard to detect for OV-3 and thus have a cubic lower bound, when $\NumThreads=\Omega(\NumEvents)$.
This establishes a complexity separation from $\HB$ races which are known to be less expressive.
Third, we show that \emph{lock-cover races} are hard for 2-OV, and thus have a quadratic lower-bound, even when $\NumThreads=2$ and $\NumLocks = \omega(\log \NumEvents)$.
The similar notion of \emph{lock-set races} is known to be detectable in $O(\NumEvents\cdot \NumLocks)$ time,
and thus we achieve a complexity separation between the two.
Moreover, we show that lock-set races become hitting-set (HS)-hard when $\NumLocks=\Theta(\NumEvents)$,
and thus also have a quadratic lower bound, when the input is sufficiently complex.
To our knowledge, this is the first work that characterizes the complexity of well-established dynamic race-detection techniques, allowing for a rigorous comparison between them.
\end{abstract}

\section{Introduction}

Concurrent programs that communicate over shared memory are prone to \emph{data races}.
Two events are \emph{conflicting} if they access the same memory location and one (at least) modifies that location.
Data races occur when conflicting access happen concurrently between different threads, and form one of the most common bugs in concurrency.
In particular, data races are often symptomatic of bugs in software like data
corruption~\cite{boehmbenign2011,racemob2013,Narayanasamy2007}, and they have been deemed \emph{pure  evil}~\cite{evil2012} due to the problems they have caused in the past \cite{SoftwareErrors2009}.
Moreover, many compiler optimizations are unsound in the presence of data races~\cite{Sevcik2008,Sevcik2011}, 
while data-race freeness is often a requirement for assigning well-defined semantics to programs~\cite{boehmadvec++2008}.

The importance of data races in concurrency has led to a multitude of techniques for detecting them efficiently~\cite{Banerjee06,vonPraun2011}.
By far the most standard approach is via \emph{dynamic analyses}.
Instead of analyzing the full program, dynamic analyzers try to \emph{predict} the existence of data races by observing and analyzing concurrent executions~\cite{Smaragdakis12,Kini17,Pavlogiannis2020}.
As full dynamic data race prediction is $\NPhard$ in general~\cite{Mathur2020b},
researchers have developed several approximate notions of dynamic races,
accompanied by efficient techniques for detecting each notion.

\SubParagraph{Happens-before races.}
The most common technique for detecting data races dynamically is based on 
Lamport's \emph{happens-before ($\HB$)} partial order~\cite{Lamport78}.
Two conflicting events form an $\HB$ race if they are unordered by $\HB$, as the lack of ordering between them indicates the fact that they may execute concurrently, thereby forming a data race.
The standard approach to $\HB$ race detection is via the use of vector clocks~\cite{djit1999},
and has seen wide success in commercial race detectors~\cite{threadsanitizer}.
As vector clock computation is known to require $\Theta(\NumEvents\cdot \NumThreads)$ time on traces of $\NumEvents$ events and $\NumThreads$ threads~\cite{CharronBost1991},
$\HB$ race detection is often assumed to suffer the same bound,
and has thus been a subject of further practical optimizations~\cite{Pozniansky03,Flanagan09}.

\SubParagraph{Synchronization preserving races.}
HB races were recently generalized to sync(hronization)-preserving races~\cite{Mathur21}.
Intuitively, two conflicting events are in a sync-preserving race if the observed trace can be soundly reordered to a witness trace in which the two events are concurrent, but without reordering synchronization events (e.g., locking events).
Similar to HB races, sync-preserving races can be detected in linear time when the number of threads is constant. However, the dependence on the number of threads is cubic for sync-preserving races, as opposed to the linear dependence for HB races.
On the other hand, sync-preserving races are known to offer better precision in program analysis.

\SubParagraph{Races based on the locking discipline.}
The locking discipline dictates that threads that access a common memory location
must do so inside \emph{critical sections}, using a common lock, when performing the access~\cite{vonPraun2011}.
Although this discipline is typically not enforced, it is considered good practice, and hence instances that violate this principle are often considered indicators of erroneous behavior.
For this reason, there have been two popular notions of data races based on the locking discipline,  
namely \emph{lock-cover races}~\cite{Dinning91} and \emph{lock-set races}~\cite{Savage97}.
Both notions are detectable in linear time when the number of locks is constant, 
however, lock-set race detection is typically faster in practice, which also comes at the cost of being less precise.

Observe that, although techniques for all aforementioned notions of races are generally thought to operate in linear time, they only do so assuming certain parameters, such as the number of threads, are constant.
However, as these techniques are deployed in runtime, often with extremely long execution traces,
they have to be as efficient as absolutely possible, often in scenarios when these parameters are very large.
When a data-race detection technique is too slow for a given application, the developers face a dilemma:~do they look for a faster algorithm, or for a simpler abstraction (i.e., a different notion of dynamic races)?
For these reasons, it is important to understand the \emph{fine-grained} complexity
of the problem at hand with respect to such parameters.
Fine-grained lower bounds can rule out the possibility of faster algorithms, 
and thus help the developers focus on new abstractions that are more tractable for the given application.
Motivated by such questions, in this work we settle the fine-grained complexity of dynamically detecting several popular notions of data races.


\subsection{Our Contributions}\label{subse:contributions}

Here we give a full account of the main results of this work, while we refer to later sections for precise definitions and proofs.
We also refer to \cref{sec:fine_grained} for relevant notions in fine-grained complexity and popular hypotheses.
The input is always a concurrent trace  $\tr$ of length $\NumEvents$, consisting of $\NumThreads$ threads, $\NumLocks$ locks, and $\NumVariables$ variables.

\Paragraph{Happens-before races.}
We first study the fine-grained complexity of $\HB$ races, as they form the most popular class of dynamic data races.
The task of most techniques is to report all events in $\tr$ that participate in an $\HB$ race, which is known to take $O(\NumEvents\cdot \NumThreads)$ time~\cite{djit1999}.
Note that the bound is quadratic when $\NumThreads=\Theta(\NumEvents)$,
and multiple heuristics have been developed to address it in practice (see e.g.,~\cite{Flanagan09}).
Our first result shows that polynomial improvements below this quadratic bound are unlikely.

\begin{restatable}{theorem}{hbovhard}\thmlabel{hb-ov-hard}
For any $\epsilon>0$, there is no algorithm that detects even a single $\HB$ race that involves a read in time $O(\NumEvents^{2-\epsilon})$, unless the $\orthv$ hypothesis fails.
\end{restatable}

Orthogonal vectors ($\orthv$) is a well-studied problem with a long-standing quadratic worst-case upper bound.
The associated hypothesis states that there is no sub-quadratic algorithm for the problem~\cite{Williams18}.
It is also known that the strong exponential time hypothesis (\seth) implies the Orthogonal Vectors hypothesis~\cite{Williams05}.
Thus, under the $\orthv$ hypothesis, \thmref{hb-ov-hard} establishes a quadratic lower bound for $\HB$ race detection.

Note that the hardness of \thmref{hb-ov-hard} arises out of the requirement to detect $\HB$ races that involve a read.
A natural follow-up question is whether detecting if the input contains \emph{any} $\HB$ race (i.e., not necessarily involving a read) has a similar lower bound based on \seth.
Our next theorem shows that under the non-deterministic \seth~(\nseth)~\cite{CarmosinoGIMPS2016},
there is no fine-grained reduction from \seth{} that proves any lower bound for this problem above $\NumEvents^{3/2}$.

\begin{restatable}{theorem}{hbdecidenosethbetterthanthreeovertwo}\thmlabel{hb-decide-no-seth-better-than-3/2}
For any $\epsilon>0$, there is no $(2^{\NumEvents}, \NumEvents^{3/2+\epsilon})$-fine-grained reduction from $\SAT$ to the problem of detecting any $\HB$ race with bound, unless \nseth{} fails.
\end{restatable}

Given the impossibility of \thmref{hb-decide-no-seth-better-than-3/2}, it would be desirable to at least show a super-linear lower bound for detecting any $\HB$ data race.
To tackle this question, we show that detecting any $\HB$ race is hard for the general problem of model checking first-order formulas quantified by $\forall\exists\exists$ on structures of size $n$ with $m$ relational tuples (denoted \fofee{}).

\begin{restatable}{theorem}{hbdecidemchard}\thmlabel{hb-decide-mc-hard}
For any $\epsilon>0$, if there is an algorithm for detecting any $\HB$ race in time $O(\NumEvents^{1+\epsilon})$,
then there is an algorithm for \fofee{} formulas in time $O(m^{1+\epsilon})$.
\end{restatable}

It is known that \fofee{} can be solved in $O(m^{3/2})$ time~\cite{Gao2018}, which yields a bound $O(n^3)$ for dense structures (i.e., when $m=\Theta(n^2)$). 
\thmref{hb-decide-mc-hard} implies that if $m^{3/2}$ is the best possible bound for \fofee{}, then
detecting any $\HB$ race cannot take $O(\NumEvents^{1+\epsilon})$ time for any $\epsilon <1/2$.
Although improvements for \fofee{} over the current $O(m^{3/2})$ bound might be possible, we find that a truly linear bound $O(m)$  would require major breakthroughs~\footnote{Even the well-studied problem of testing triangle freeness, which is a special case of the similarly flavored \foeee, has the super-linear bound $O(n^{\omega})$.}.
Under this hypothesis, \thmref{hb-decide-mc-hard} implies a super-linear bound for $\HB$ races.

Finally, we give an improved upper bound for this problem when $\NumLocks=o(\NumThreads)$.

\begin{restatable}{theorem}{hbalgo}\thmlabel{hb-algo}
Deciding whether $\tr$ has an $\HB$ race can be done in time $O(\NumEvents\cdot \min(\NumThreads, \NumLocks))$.
\end{restatable}

In fact, similar to existing techniques~\cite{Flanagan09}, the algorithm behind \thmref{hb-algo} detects \emph{all}
variables that participate in an $\HB$ race (instead of just reporting $\tr$ as racy).


\Paragraph{Synchronization-preserving races.}
Next, we turn our attention to the recently introduced sync-preserving races~\cite{Mathur2020b}.
It is known that detecting sync-preserving races takes $O(\NumEvents\cdot\NumVariables\cdot  \NumThreads^3)$ time.
As sync-preserving races are known to be more expressive than $\HB$ races,
the natural question is whether sync-preserving races can be detected more efficiently, e.g., by an algorithm that achieves a bound similar to \thmref{hb-algo} for $\HB$ races.
Our next theorem answers this question in negative.

\begin{restatable}{theorem}{syncpovthreehard}\thmlabel{syncp-ov3-hard}
For any $\epsilon>0$, there is no algorithm that detects even a single sync-preserving race in time $O(\NumEvents^{3-\epsilon})$, unless the $\orthvk{3}$ hypothesis fails.
Moreover, the statement holds even for traces over a single variable.
\end{restatable}

As $\HB$ races take at most quadratic time, \thmref{syncp-ov3-hard} shows that the increased expressiveness of sync-preserving races incurs a complexity overhead that is unavoidable in general.

\Paragraph{Races based on the locking discipline.}
We now turn our attention to data races based on the locking discipline, namely \emph{lock-cover races}  and \emph{lock-set races}.
It is known that lock-cover races are more expressive than lock-set races.
On the other hand, existing algorithms run in $O(\NumEvents^2\cdot \NumLocks)$ time for lock-cover races and in $O(\NumEvents\cdot \NumLocks)$ time for lock-set races, and thus hint that the former are computationally harder to detect.
Our first theorem makes this separation formal, by showing that even with just two threads, having slightly more that logarithmically many locks implies a quadratic hardness for lock-cover races.

\begin{restatable}{theorem}{lockcoverquadraticlowerbound}\thmlabel{lock-cover-quadratic-lower-bound}
For any $\epsilon>0$, any $\NumThreads \geq 2$ and any $\NumLocks=\omega(\log \NumEvents)$,
there is no algorithm that detects even a single lock-cover race in time $O(\NumEvents^{2-\epsilon})$, unless the $\orthv$ hypothesis fails.
\end{restatable}

Observe that the $O(\NumEvents\cdot \NumLocks)$ bound for lock-set races also becomes quadratic, when the number of locks is unbounded (i.e., $\NumLocks =\Theta(\NumEvents))$.
Is there a \seth-based quadratic lower bound similar to \thmref{lock-cover-quadratic-lower-bound} for this case?
Our next theorem rules out this possibility, again under \nseth{}.

\begin{restatable}{theorem}{lock-set-nseth}\thmlabel{lock-set-nseth}
For any $\epsilon>0$, there is no $(2^\NumEvents,\NumEvents^{1+\epsilon})$-fine-grained reduction from $\SAT$ to the problem of detecting any lock-set race, unless \nseth{} fails.
\end{restatable}

Hence, even though we desire a quadratic lower bound, \thmref{lock-set-nseth} rules out any super-linear lower-bound based on \seth.
Alas, our next theorem shows that a quadratic lower bound for lock-set races does exist, based on the hardness of the hitting set ($\hs$) problem.

\begin{restatable}{theorem}{locksetquadraticlowerbound}\thmlabel{lock-set-quadratic-lower-bound}
For any $\epsilon>0$ and any $\NumThreads = \omega(\log n)$,
there is no algorithm that detects even a single lock-set race in time $O(\NumEvents^{2-\epsilon})$, unless the $\hs$ hypothesis fails.
\end{restatable}

Hitting set is a problem similar to $\orthv$, but has different quantifier structure.
Just like the $\orthv$ hypothesis, the $\hs$ hypothesis states that there is no sub-quadratic algorithm for the problem~\cite{Abboud2016}.
Although $\hs$ implies $\orthv$, the opposite is not known, and thus \thmref{lock-set-quadratic-lower-bound} does not contradict \thmref{lock-set-nseth}.
In conclusion, we have that both lock-cover and lock-set races have (conditional) quadratic lower bounds,
though the latter is based on a stronger hypothesis ($\hs$), and requires more threads and locks for hardness to arise.

Finally, on our way to \thmref{lock-set-nseth}, we obtain the following theorem.

\begin{restatable}{theorem}{locksetsinglevariablelinear}\thmlabel{lock-set-single-variable-linear}
Deciding whether a trace $\tr$ has a lock-set race on a variable $x$ can be performed in $O(\NumEvents)$ time.
Thus, deciding whether $\tr$ has a lock-set race can be performed in $O(\NumEvents\cdot \min(\NumLocks, \NumVariables))$ time.
\end{restatable}

Hence, \thmref{lock-set-single-variable-linear} strengthens the $O(\NumEvents \cdot \NumLocks)$ upper bound for lock-set races when $\NumVariables=o(\NumLocks)$.


\subsection{Related Work}\label{subsec:related}

\Paragraph{Dynamic data-race detection.}
There exists a rich literature in dynamic techniques for data race detection.
Methods based on vector clocks (\djit algorithm~\cite{djit1999})
using Lamport's Happens Before ($\HB$)~\cite{Lamport78}
and the lock-set principle in Eraser~\cite{Savage97} 
were the first ones to popularize dynamic analysis
for detecting data races.
Later work attempted to increase the performance of these notions
using optimizations as in~\cite{Pozniansky03} and \fasttrack~\cite{Flanagan09}, 
altogether different algorithms (e.g., the \goldilocks algorithm~\cite{Elmas07}),
and hybrid techniques~\cite{OCallahan03}.
$\HB$ and lock-set based race detection are respectively
sound (but incomplete) and complete (but unsound) variants 
of the more general problem of data-race \emph{prediction}~\cite{Koushik05}.
While earlier work on data race prediction focused on explicit~\cite{Koushik05}
or symbolic~\cite{rvpredict,Said11} enumeration, recent efforts have focused
on scalability~\cite{Smaragdakis12,Mathur18,Kini17,Pavlogiannis2020,Roemer18,Sulzmann2020}.
The more recent notion~\cite{Mathur21} of sync-preserving races
generalizes the notion of $\HB$.
As the complexity of race prediction is prohibitive ($\NPhard$ in general~\cite{Mathur2020b}), this work characterizes the fine-grained complexity of popular, more relaxed notions of dynamic races that take polynomial time.


\Paragraph{Fine-grained complexity.}
Traditional complexity theory usually shows a problem is intractable 
by proving it NP-hard, and tractable by showing it is in $\PTime$. 
For algorithms with large input sizes, this distinction may be too coarse.
It becomes important to understand, even for problems in $\PTime$, 
whether algorithms with smaller degree polynomials than the known are possible, 
or if there are fine-grained lower bounds making this unlikely. 
Fine-grained complexity involves proving such lower bounds, by showing relationships between problems in $\PTime,$ with an emphasis on the degree of the complexity polynomial,
and is nowadays a field of very active study. We refer to~\cite{Bringmann2019} for an introductory, and to~\cite{Williams18} for a more extensive exposition on the topic.
Fine-grained arguments have also been instrumental in
characterizing the complexity of various problems in concurrency, such as bounded context-switching~\cite{Chini2017}, safety verification~\cite{Chini2018}, data-race prediction~\cite{Mathur2020b} and consistency checking~\cite{Chini2020}.


\section{Preliminaries}


\subsection{Concurrent Program Executions and Data Races}

\myparagraph{Traces and Events}{
We consider execution traces (or simply \emph{traces})
generated by concurrent programs, under the sequential consistency memory model.
Under this memory model, a trace $\tr$ is a sequence of events.
Each event $e$ is labeled with a tuple $\lab(e) = \ev{t, op}$,
where $t$ is the (unique) identifier of the thread that performs the event $e$,
and $op$ is the operation performed in $e$.
We will often abuse notation and write $e = \ev{t, op}$
instead of $\lab(e) = \ev{t, op}$.
For the purpose of this presentation, an operation can be one of
\begin{enumerate*}[label=(\alph*)]
	\item read ($\rd(x)$) from, or write ($\wt(x)$) to, a shared memory variable $x$,
	\item $\acq(\lk)$ or $\rel(\lk)$ of a lock $\lk$.
\end{enumerate*}

For an event $e = \ev{t, op}$, we use $\ThreadOf{e}$ and $\OpOf{e}$ 
to denote respectively the thread identifier $t$ and the operation $op$.
For a trace $\tr$, we use $\events{\tr}$ to denote the set of events that appear
in $\tr$.
Similarly, we will use $\threads{\tr}$, $\locks{\tr}$ and $\vars{\tr}$
to denote respectively the set of threads, locks and shared variables
that appear in trace $\tr$.
We denote by $\NumEvents=|\events{\tr}|$, $\NumThreads=|\threads{\tr}|$, 
$\NumLocks=|\locks{\tr}|$, and $\NumVariables=|\vars{\tr}|$.
The set of read events and write events on variable $x \in \vars{\tr}$
will be denoted by $\reads{\tr}(x)$ and $\writes{\tr}(x)$,
and further we let $\accesses{\tr}(x) = \reads{\tr}(x) \cup \writes{\tr}(x)$.
Similarly, we let $\acquires{\tr}(\lk)$ and $\releases{\tr}(\lk)$ denote the set of lock-acquire and lock-release events, respectively, of $\tr$ on lock $\lk$.
The \emph{trace order} of $\tr$, denoted $\trord{\tr}$,
is the total order on $\events{\tr}$ induced by the sequence $\tr$.
Finally, the \emph{thread-order} of $\tr$, denoted $\tho{\tr}$
is the smallest partial order on $\events{\tr}$
such that for any two events $e_1, e_2 \in \events{\tr}$,
if $e_1 \trord{\tr} e_2$ and $\ThreadOf{e_1} = \ThreadOf{e_2}$,
then $e_1 \tho{\tr} e_2$.

Traces are assumed to be well-formed in that critical sections
on the same lock do not overlap.
For a lock $\lk \in \locks{\tr}$, let $\proj{\tr}{\lk}$ be the projection of the trace $\tr$ 
on the set of events $\setpred{e}{\OpOf{e} \in \set{\acq(\lk), \rel(\lk)}}$.
Also, let $t_1, \ldots t_k$ be the thread identifiers in $\threads{\tr}$.
Well-formedness then entails that for each lock $\lk$,
the projection $\proj{\tr}{\lk}$ is a prefix of some string
in the language of the grammar with production rules
$S \rightarrow \varepsilon | S \cdot S_{t_1} | S \cdot S_{t_2} | \cdots | S \cdot S_{t_k}$
and
$S_{t_i} \rightarrow \ev{t_i, \acq(\lk)} \cdot \ev{t_i, \rel(\lk)}$
and start symbol $S$.
Thus, every release event $e$ has a unique matching acquire event, which
we denote by $\match{\tr}(e)$.
Likewise for an acquire event $e$, $\match{\tr}(e)$ denotes the unique
matching release event if one exists.
For an acquire event $e$, the critical section of $e$
is the set of events $\crit{\tr}(e)=\setpred{f}{e \tho{\tr} f \tho{\tr} \match{\tr}(e)}$
if $\match{\tr}(e)$ exists, and $\crit{\tr}(e)=\setpred{f}{e \tho{\tr} f}$ otherwise.
}

\myparagraph{Data Races}{
	Two events $e_1, e_2 \in \events{\tr}$ are said to be \emph{conflicting}
	if they are performed by different threads,
	they are access events touching the same memory location, and at least
	one of them is a write access.
	Formally, we have (i)~$\ThreadOf{e_1} \neq \ThreadOf{e_2}$,
	(ii)~$e_1, e_2\in \accesses{\tr}(x)$ for some $x\in \vars{\tr}$, and
	(iii)~$\{e_1, e_2\}\cap \writes{\tr}(x)\neq \emptyset$.	
	An event $e \in \events{\tr}$ is said to be $\emph{enabled}$ in a prefix
	$\rho$ of $\tr$, if for every event $e' \neq e$ with $e' \tho{\tr} e$,
	we have $e' \in \events{\rho}$.
	A data race in $\tr$ is a pair of conflicting events $(e_1, e_2)$
	such that there is a prefix $\rho$ in which both $e_1$ and $e_2$
	are simultaneously enabled.
	Depending on the type of access of $e_1$ and $e_2$,
	we often distinguish between write-write races and write-read races.
	}


\subsection{Notions of Dynamic Data Races}

As the problem of determining whether a concurrent program has an execution with a data race is undecidable, 
dynamic techniques observe program traces and report whether certain events indicate the presence of a race.
Depending on the technique, such reports can be sound (i.e., they guarantee the presence of a race in the program),
Here we describe in detail some popular approaches to dynamic race detection that are the subject of this work.

\Paragraph{Happens-Before Races.}
Given a trace $\tr$, the \emph{happens before} order $\hb{\tr}$ is the smallest partial order on $\events{\tr}$ such that
\begin{enumerate*}[(a)]
	\item $\tho{\tr} \subseteq \hb{\tr}$, and
	\item for any lock $\lk \in \locks{\tr}$ and for events $e \in \acquires{\tr}(\lk)$
	and $f \in \releases{\tr}(\lk)$,
	if $e \trord{\tr} f$ then $e \hb{\tr} f$.
\end{enumerate*}
A pair of conflicting events $(e_1, e_2)$
is an $\HB$-race in $\tr$ if they are unordered by $\HB$, i.e., 
$e_1 \not\hb{\tr} e_2$ and $e_2 \not\hb{\tr} e_1$.
The associated decision question is, 
\emph{given a trace $\tr$, determine whether $\tr$ has an $\HB$ race}.
Typically $\HB$ race detectors are tasked to report all events that form $\HB$ race with an earlier event in the trace~\cite{threadsanitizer,intel-inspector,helgrind}).
That is, they solve the following function problem:\emph{given a trace $\tr$, determine 
all events $e_2\in \events{\tr}$ for which there exists an event $e_1\in \events{\tr}$ such that
$e_1 \trord{\tr} e_2$, and $(e_1, e_2)$ is an $\HB$ race of $\tr$}.
The standard algorithm for solving both versions of the problem is a vector-clock algorithm
that runs in $O(\NumEvents\cdot \NumThreads)$ time~\cite{djit1999}.

\Paragraph{Synchronization Preserving Races.}
Next, we present the notion of \emph{sync(hronization)-preserving races}~\cite{Mathur2020b}.
For a trace $\tr$ and a read event $e$,
we use $\lw{\tr}(e)$ to denote the write event observed by $e$.
That is, $e' = \lw{\tr}(e)$ is the last (according to the trace order $\trord{\tr}$) 
write event $e'$ of $\tr$ such that $e$ and $e'$ access the same variable and $e' \trord{\tr} e$; 
if no such $e'$ exists, then we write $\lw{\tr}(e) = \bot$.
A trace $\rho$ is said to be a correct reordering of trace $\tr$, if 
\begin{enumerate*}[label=(\alph*)]
\item $\events{\rho} \subseteq \events{\tr}$
\item $\events{\rho}$ is downward closed with respect to $\tho{\tr}$, 
and further $\tho{\rho} \subseteq \tho{\tr}$, and
\item for every read event $e \in \events{\rho}$, 
$\lw{\rho}(e) = \lw{\tr}(e)$.
\end{enumerate*}
We say that $\rho$ is \emph{sync-preserving} with respect to
$\tr$ if for every lock $\lk$ and for any two acquire events 
$e_1, e_2 \in \acquires{\rho}(\lk)$,
we have $e_1 \trord{\rho} e_2$ iff $e_1 \trord{\tr} e_2$.
That is, the order of two critical sections on the same
lock is the same in $\tr$ and  $\rho$.

A pair of conflicting events $(e_1, e_2)$ is a \emph{sync-preserving race} in $\tr$ if $\tr$ has a sync-preserving correct reordering $\rho$ such that $(e_1, e_2)$ is a data race of $\rho$.
The associated decision question is, 
\emph{given a trace $\tr$, determine whether $\tr$ has a sync-preserving race}.
As with $\HB$ races, we are typically interested in reporting  
all events $e_2\in \events{\tr}$ for which there exists an event $e_1\in \events{\tr}$ 
such that $e_1 \trord{\tr} e_2$, and $(e_1, e_2)$ is an $\HB$ race of $\tr$.
It is known one can report all such events $e_2$ in time 
$O(\NumEvents\cdot \NumVariables \cdot \NumThreads^3)$.

\Paragraph{Lock-Cover and Lock-Set Races.}
Lock-cover and lock-set races indicate violations of the \emph{locking discipline}.
For an event $e$ in a trace $\tr$, let 
$\lheld{\tr}(e) = \setpred{\lk}{ \exists f \in \acquires{\tr}(\lk), \text{ such that } e \in \crit{\tr}(f)}$,
i.e., $\lheld{\tr}(e)$ is the set of locks held by thread $\ThreadOf{e}$ when $e$ is executed.
A pair of conflicting events might indicate a data race if $\lheld{\tr}(e_1)\cap \lheld{\tr}(e_2)=\emptyset$.
Although this condition does not guarantee the presence of a race, it constitutes a violation of the locking discipline and can be further investigated.

A pair of conflicting events $(e_1, e_2)$ is a \emph{lock-cover race} if $\lheld{\tr}(e_1)\cap \lheld{\tr}(e_2)=\emptyset$.
The decision question is, \emph{given a trace $\tr$, determine if $\tr$ has a lock-cover race}.
The problem is solvable in $O(\NumEvents^2\cdot \NumLocks)$ time, by checking the above condition over all conflicting event pairs.

As the algorithm for lock-cover races takes quadratic time, 
developers often look for less expensive indications of violations of locking discipline, called lock-set races (as proposed by \eraser race detector~\cite{Savage97}).
A trace $\tr$ has a \emph{lock-set race} on variable $x\in \vars{\tr}$ if
\begin{compactenum}[(a)]
\item there exists a pair of conflicting events $(e_1, e_2)\in \writes{\tr}(x)\times \accesses{\tr}(x)$, and
\item $\bigcap_{e \in \accesses{\tr}(x)} \lheld{\tr}(e) = \emptyset$.
\end{compactenum}
The associated decision question is, \emph{given a trace $\tr$, determine if $\tr$ has a lock-set race}.
Note that a lock-cover race implies a lock-set race, but not the other way around.
On the other hand, determining whether $\tr$ has a lock-set race is easily performed in $O(\NumEvents\cdot \NumLocks)$ time.


\begin{figure}[t]
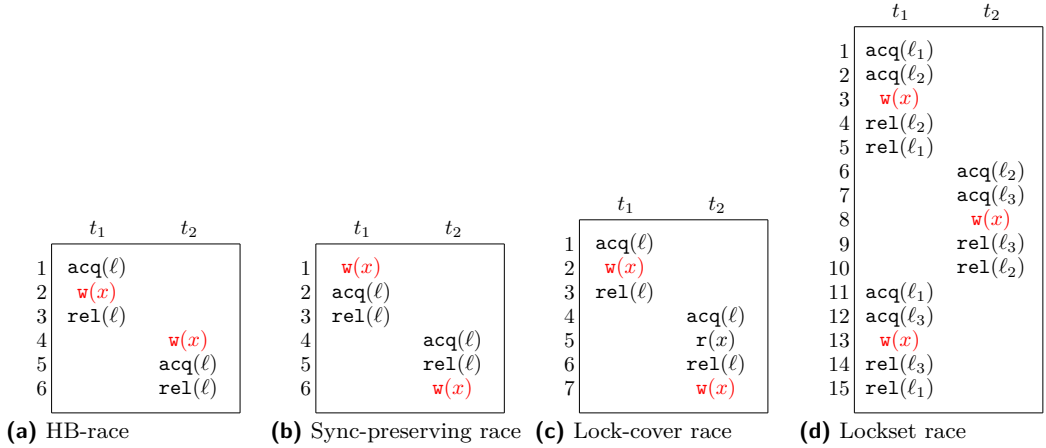

\begin{subfigure}[b]{0.24\textwidth}
\centering
\scalebox{0.8}{
\execution{2}{
	\figev{1}{$\acq(\lk)$}
	\figev{1}{$\color{red} \wt(x)$}
	\figev{1}{$\rel(\lk)$}
	\figev{2}{$\color{red} \wt(x)$}
	\figev{2}{$\acq(\lk)$}
	\figev{2}{$\rel(\lk)$}
}
}
\end{subfigure}
\begin{subfigure}[b]{0.24\textwidth}
\centering
\scalebox{0.8}{
\execution{2}{
	\figev{1}{$\color{red} \wt(x)$}
	\figev{1}{$\acq(\lk)$}
	\figev{1}{$\rel(\lk)$}
	\figev{2}{$\acq(\lk)$}
	\figev{2}{$\rel(\lk)$}
	\figev{2}{$\color{red} \wt(x)$}
}
}
\end{subfigure}
\begin{subfigure}[b]{0.24\textwidth}
\scalebox{0.8}{
\execution{2}{
	\figev{1}{$\acq(\lk)$}
	\figev{1}{$\color{red} \wt(x)$}
	\figev{1}{$\rel(\lk)$}
	\figev{2}{$\acq(\lk)$}
	\figev{2}{$\rd(x)$}
	\figev{2}{$\rel(\lk)$}
	\figev{2}{$\color{red} \wt(x)$}
}
}
\end{subfigure}
\begin{subfigure}[b]{0.24\textwidth}
\scalebox{0.8}{
\execution{2}{
	\figev{1}{$\acq(\lk_1)$}
	\figev{1}{$\acq(\lk_2)$}
	\figev{1}{$\color{red} \wt(x)$}
	\figev{1}{$\rel(\lk_2)$}
	\figev{1}{$\rel(\lk_1)$}
	\figev{2}{$\acq(\lk_2)$}
	\figev{2}{$\acq(\lk_3)$}
	\figev{2}{$\color{red} \wt(x)$}
	\figev{2}{$\rel(\lk_3)$}
	\figev{2}{$\rel(\lk_2)$}
	\figev{1}{$\acq(\lk_1)$}
	\figev{1}{$\acq(\lk_3)$}
	\figev{1}{$\color{red} \wt(x)$}
	\figev{1}{$\rel(\lk_3)$}
	\figev{1}{$\rel(\lk_1)$}
}
}
\end{subfigure}\\
\begin{subfigure}{0.24\textwidth}
\caption{HB-race}
\figlabel{ex-hb}
\end{subfigure}
\begin{subfigure}{0.24\textwidth}
\caption{Sync-preserving race}
\figlabel{ex-syncp}
\end{subfigure}
\begin{subfigure}{0.24\textwidth}
\caption{Lock-cover race}
\figlabel{ex-lock-cover}
\end{subfigure}
\begin{subfigure}{0.24\textwidth}
\caption{Lockset race}
\figlabel{ex-lockset}
\end{subfigure}
\caption{Types of data races.}
\figlabel{ex-races}
\end{figure}

\Paragraph{Example.}
We illustrate the different notions of races in \figref{ex-races}.
We use $e_i$ to denote the $i^\text{th}$ event of the trace in consideration.
First consider the trace $\tr_a$ in \figref{ex-hb}.
The events $e_2$ and $e_4$
are conflicting and unordered by $\hb{\tr_a}$, thus $(e_2, e_4)$ is an $\HB$-race.
Second, in trace $\tr_b$  of \figref{ex-syncp}, the pair $(e_1, e_6)$ is not an $\HB$-race
as $e_1 \hb{\tr_b} e_6$.
But this is a sync-preserving race witnessed by the correct reordering $e_4,e_5$, as both $e_1$ and $e_6$ are enabled.
Third, in trace $\tr_c$ of \figref{ex-lock-cover}, the pair $(e_2, e_7)$ is neither a sync-preserving race nor an $\HB$ race, but is a lock-cover race as $\lheld{\tr_c}(e_2) \cap \lheld{\tr_c}(e_7) = \emptyset$.
Finally, the trace $\tr_d$ in \figref{ex-lockset} has no $\HB$,
sync-preserving or lock-cover race, as all $\wt(x)$ are protected by a common lock.
But there is a lock-set race on $x$ as there is no single lock that protects all $\wt(x)$.

\section{Happens-Before Races}
\seclabel{hb}

In this section we prove the results for detecting $\HB$ races, i.e., \thmref{hb-ov-hard}~to \thmref{hb-algo}.


\subsection{Algorithm for $\HB$ Races}
\seclabel{hb-upper-bound}

In this section, we outline our $O(\NumEvents \cdot \NumLocks)$-time
algorithm for checking if  a trace $\tr$ has an $\HB$-race,
thereby proving \thmref{hb-algo}.
As with the standard vector clock algorithm~\cite{djit1999}, 
our algorithm is based on computing timestamps for each event.
However, unlike the standard algorithm that assigns thread-indexed timestamps, 
we use \emph{lock-indexed} timestamps, or \emph{lockstamps}, which we formalize next.
We fix the input trace $\tr$ in the rest of the discussion.

\Paragraph{Lockstamps}{
A lockstamp is a mapping from locks to natural numbers 
(including infinity) $L : \locks{\tr} \to \Natsinf$.
Given lockstamps $L, L_1, L_2$ and lock $\lk$, we use
the notation 
\begin{enumerate*}[label=(\roman*)]
\item $L[\lk \mapsto c]$ to denote the the lockstamp $\lambda m {\cdot} \text{ if } m = \lk \text{ then } c \text{ else } L(m)$,
\item $L_1 \mx L_2$ to denote the pointwise maximum, i.e., $(L_1\mx L_2)(\lk) = \max(L_1(\lk), L_2(\lk))$ for every $\lk$,
\item $L_1 \mn L_2$ to denote the pointwise minimum,
and 
\item $L_1 \cle L_2$ to denote the predicate $\forall \lk {\cdot} L_1(\lk) \leq L_2(\lk)$.
\end{enumerate*}
}

Our algorithm computes \emph{acquire} and \emph{release} lockstamps $\acqls{\tr}_e$
and $\rells{\tr}_e$ for every event $e \in \events{\tr}$.
Let us formalize these next.
For a lock $\lk$ and acquire event $f \in \acquires{\tr}(\lk)$
(resp. release event $g \in \releases{\tr}(\lk)$), 
let $pos_\tr(f) = |\setpred{f' \in \acquires{\tr}(\lk)}{f' \trord{\tr} f}|$ 
(resp. $pos_\tr(g) = |\setpred{f' \in \releases{\tr}(\lk)}{f' \trord{\tr} f}|$) 
denote the relative position of $f$ (resp. $g$)
among all acquire events (resp. release events) of $\lk$.
Then, for an event $e \in \events{\tr}$ 
the lockstamps $\acqls{\tr}_e$ and $\rells{\tr}_e$
are defined as follows (we assume that $\max\emptyset = 0$ and $\min\emptyset = \infty$.)
\begin{equation}
\begin{array}{rcl}
\acqls{\tr}_e(\lk) &=& \lambda \lk \cdot \max\setpred{pos_\tr(f)}{f \in \acquires{\tr}(\lk), f \hb{\tr} e} \\
\quad \\
\rells{\tr}_e(\lk) &=& \lambda \lk \cdot \min\setpred{pos_\tr(f)}{f \in \releases{\tr}(\lk), e \hb{\tr} f}
\end{array}
\end{equation}

Our $O(\NumEvents \cdot \NumLocks)$ algorithm now
relies on the following observations. 
First, the $\HB$ partial order can be inferred by comparing lockstamps of events (\lemref{hb-isomorphic-lockstamp}).
Second, there is an $O(\NumEvents \cdot \NumLocks)$ time algorithm
that computes the acquire and release lockstamps for each event in the input trace.
Third, the existence of an $\HB$ race can be determined
by examining only $O(\NumEvents)$ pairs of conflicting events (using their lockstamps),
instead of all possible $O(\NumEvents^2)$ pairs 
(\lemref{num-consecutive}).
Finally, we can also examine all the $O(\NumEvents)$ pairs in
time $O(\NumEvents \cdot \NumLocks)$ (using $O(\NumEvents)$ lockstamp comparisons)
and thus determine the existence of an $\HB$ race in the same asymptotic running time.
Let us first state how we use lockstamps to infer the $\HB$ relation.

\begin{restatable}{lemma}{hbisomorphiclockstamp}\lemlabel{hb-isomorphic-lockstamp}
Let $e_1 \trord{\tr} e_2$ be events in $\tr$ such that $\ThreadOf{e_1} \neq \ThreadOf{e_2}$.
We have, $e_1 \hb{\tr} e_2 \iff \neg (\acqls{\tr}_{e_2} \cle \rells{\tr}_{e_1})$
\end{restatable}

\myparagraph{Computing Lockstamps}{
We now illustrate how to compute the acquire 
lockstamps for all events, by processing the trace $\tr$ in a forward  pass. 
For each thread $t$ and lock $\lk$, we maintain lockstamp variables $\CTC_t$ and $\CTL_\lk$.
We also maintain an integer variable $\idxlk_\lk$
for each lock $\lk$ that stores the index of 
the latest $\acq(\lk)$ event in $\tr$.
Initially, we set each $\CTC_t$ and $\CTL_m$ to the \emph{bottom}
map $\lambda \lk \cdot 0$, and $\idxlk_m$ to $0$, for each thread $t$
and lock $m$.
We traverse $\tr$ left to right,
and perform updates to the data structures as described in \algoref{assign-acquire-ls},
by invoking the appropriate \emph{handler}
based on the thread and operation of the current event $e = \ev{t, op}$.
At the end of each handler, we assign the lockstamp $\acqls{\tr}_e$ to $e$.
The computation of release
lockstamps is similar, albeit in a reverse pass, and presented in \appref{proofs_hb_ub}.
Observe that each step takes $O(\NumLocks)$ time giving us a total
running time of $O(\NumEvents \cdot \NumLocks)$ to assign lockstamps.
}


{
	\footnotesize
\begin{algorithm*}[h]
\begin{multicols}{3}

\myhandler{\acqhandler{$t$, $\lk$}}{
	$\idxlk_\lk\gets \idxlk_\lk + 1$ \;
	$\CTC_t\gets \CTC_t[\lk \mapsto \idxlk_\lk] \mx\CTL_\lk$ \;
	$\acqls{\tr}_e \gets \CTC_t$
}

\myhandler{\relhandler{$t$, $\lk$}}{
	$\CTL_\lk\gets \CTC_t$ \;
	$\acqls{\tr}_e \gets \CTC_t$
}

\myhandler{\rdhandler{$t$, $x$}}{
	$\acqls{\tr}_e \gets \CTC_t$
}

\BlankLine

\myhandler{\wthandler{$t$, $x$}}{
	$\acqls{\tr}_e \gets \CTC_t$
}

\end{multicols}
\normalsize
\caption{\textit{Assigning acquire lockstamps to events in the trace}}
\algolabel{assign-acquire-ls}
\SetAlgoInsideSkip{medskip}
\end{algorithm*}
\normalsize
}

We say that a pair of conflicting access events $(e_1, e_2)$ (with $e_1 \trord{\tr} e_2$)
to a variable $x$ is a \emph{consecutive conflicting pair} if there is no 
event $f \in \writes{\tr}(x)$ such that $e_1 \stricttrord{\tr} f \stricttrord{\tr} e_2$.
We make the following observation.


\begin{restatable}{lemma}{numconsecutive}\lemlabel{num-consecutive}
A trace $\tr$ has an $\HB$-race iff there is pair of consecutive conflicting events in $\tr$ that is an $\HB$-race.
Moreover, $\tr$ has at most $O(\NumEvents)$ many consecutive conflicting pairs of events.
\end{restatable}

\myparagraph{Checking for an $\HB$ race}{
We now describe the algorithm for checking for an $\HB$ race in $\tr$.
We perform a forward pass on $\tr$ while
storing the release lockstamps of some of the earlier events.
When processing an access event $e$, we check if it is in race with
an earlier event by comparing the acquire lockstamp of $e$
with a previously stored release lockstamp.
More precisely, we maintain a variable $\VCW_x$ to store the release lockstamp
of the last write event on $x$, a variable $t^w_x$ to store the thread 
that performed this write
and set $\setrdls_x$ to store pairs $(t, L)$
of threads and release lockstamps of all the read
events performed since the last write on $x$ was observed.
Initially, $t^w_x = \texttt{NIL}$, $\VCW_x=\lambda \lk \cdot \infty$
and $\setrdls_x = \emptyset$.
The update performed at each event $e = \ev{t, op}$ are presented in 
the corresponding handler in \algoref{check-hb-race}.


{
\footnotesize
\begin{algorithm*}[h]
\begin{multicols}{2}
\myhandler{\rdhandler{$t$, $x$}}{
	\If{$t^w_x \not\in \set{\texttt{NIL}, t} \land \acqls{\tr}_e \cle \VCW_x$}{
		\declare `race' and \exit
	}
	$\setrdls_x \gets \setrdls_x \cup \set{(t, \rells{\tr}_e)}$
}

\myhandler{\wthandler{$t$, $x$}}{
	\If{$t^w_x \not\in \set{\texttt{NIL}, t} \land \acqls{\tr}_e \cle \VCW_x$}{
		\declare `race' and \exit
	}
	\If{$\exists (u, L) \in \setrdls_x, t \neq u \land \acqls{\tr}_e \cle L$}{
	\declare `race' and \exit
	}
	$t^w_x = t$;
	$\setrdls_x \gets \emptyset$; $\VCW_x \gets \rells{\tr}_e$
}

\end{multicols}
\normalsize
\caption{\textit{Determining the existence of an $\HB$-race using lockstamps}}
\algolabel{check-hb-race}
\SetAlgoInsideSkip{medskip}
\end{algorithm*}
\normalsize
}

}

We refer to \appref{proofs_hb_ub} for the correctness, which concludes the proof of \thmref{hb-algo}.

\subsection{Hardness Results for $\HB$}
\seclabel{hb-lower-bound}

We now turn our attention to the hardness results for $\HB$ race detection.
To this end, we prove \thmref{hb-ov-hard}, \thmref{hb-decide-no-seth-better-than-3/2}, and \thmref{hb-decide-mc-hard}.
We start with defining the graph $\hbgraph$, which can be thought of as a form of transitive reduction of the $\HB$ relation.

\myparagraph{The graph $\hbgraph$}{  
Given a trace $\tr$, the graph $\hbgraph$ is a graph with node set $\events{\tr}$, and we have
an edge $(e_1, e_2)$ in $\hbgraph$ iff
\begin{enumerate*}[label=(\roman*)]
\item $e_2$ is the immediate successor of $e_1$ wrt the thread order $\tho{\tr}$, or
\item $e_1$ is a $\rel(\lk)$ event, $e_2$ is a $\acq(\lk)$ event, $e_1\trord{\tr}e_2$, and there is no intermediate event in $\tr$ that accesses lock $\lk$.
\end{enumerate*}
It follows easily that for any two distinct events $e_1, e_2$, we have $e_1\hb{\tr}e_2$ iff $e_2$ is reachable from $e_1$ in $\hbgraph$.
Moreover, every node has out-degree $\leq 2$ and thus $\hbgraph$ is sparse, while it can be easily constructed in $O(\NumEvents)$ time.
}


\begin{figure}
\begin{subfigure}{.2\textwidth}
\scalebox{0.8}{
\OVTwo{
\ovX{101}
\ovX{100}
\ovX{010}
}{
\ovY{111}
\ovY{011}
\ovY{110}
}{
}
}
\caption*{\orthv{} instance} 
\end{subfigure}
\begin{subfigure}{.5\textwidth}
\scalebox{0.8}{
\leanExecIndex{(1,0)}{
	\figev{1}{$w(z)$}
	\figev{1}{$\critsec(\lk^x)$}
}
\leanExecIndex{(1,1)}{
	\figevoffset{3}{1}{$\critsec(\lk^x)$}
	\figevoffset{4}{1}{$\critsec(\lk_1)$}
}
\leanExecIndex{(1,2)}{
	\figevoffset{7}{1}{$\critsec(\lk^x)$}
}
}
\end{subfigure}
\begin{subfigure}{.2\textwidth}
\scalebox{0.8}{
\leanExecIndex{(1,3)}{
	\figevoffset{9}{1}{$\critsec(\lk^x)$}    \figevoffset{10}{1}{$\critsec(\lk_3)$}
	\figevoffset{11}{1}{$\critsec(\lk_{(y_1,3)})$}
	\figevoffset{12}{1}{$\critsec(\lk_{(y_2,3)})$}
	\figevoffset{13}{1}{$\critsec(\lk_{(y_3,3)})$}
}
}
\end{subfigure}
\caption{Reducing \orthv{} to detecting a write-read $\HB$-races.
Illustration of the threads $t(x,i)$, where $x$ is the first vector of $A_1$.
$\critsec(\ell)$ denotes the sequence $\acq(\ell),\rel(\ell)$.
Event numbers indicate the relative order in which these threads execute in $\tr$.
}
\figlabel{ov-hb-example}
\end{figure}

\Paragraph{\orthv{} hardness of write-read $\HB$ races.}
Given a \orthv{} instance \ovinstwo{} on two vector sets $A_1, A_2$, we create a trace $\tr$ as follows. 
For the part $A_1$ of \orthv{}, we introduce $n\cdot (d+1)$ threads denoted by $t(x,i)$, for $x\in[n], i\in \{0\}\cup [d],$ and $d$ locks, each denoted by $l_i$, for $i\in[d]$.
For the second part $A_2$ we introduce $n\cdot d$ locks denoted by $l(y,i)$, for $y\in [n], i\in [d],$ and $n$ threads, denoted by $t_y$, for $y\in [n]$. 
Finally, we have a single variable $z$.

We first describe the threads $t(x,i)$.
We order the vectors in $A_1$ arbitrarily. 
For each vector $x$, for each $i\in [d]$ with $x[i]=1$, we introduce a critical section on the lock $l_i$. 
If $x$ is the last vector of $A_1$ with $x[i]=1,$ 
we also insert the critical sections $l_{(y,i)}$ for all  $y\in [n]$, to $t(x,i)$ after the critical section of $l_x$. 
Finally, we construct a thread $t_{x,0}$ which starts with a write event $\wt(z)$, 
followed by a critical section on lock $l^x$.
We also insert a critical section on lock $l^x$ to all threads $t(x,i)$, for $i\in [d]$.
Hence the $\wt(z)$ event is ordered by $\HB$ before all other events of $t(x,i)$.
See \figref{ov-hb-example} for an illustration.

Now we describe the threads $t_y$. 
For each $i\in [d],$ if $y[i]=1$, we add a critical section of the lock $l(y,i)$ in $t_y$. 
We end the thread with a read event $\rd(z)$.

Finally, we construct $\tr$ by first executing each thread $t(x,i)$ in the pre-determined order of $x\in A_1$, followed by executing the traces $t_y$ in any order.
See  \figref{ov-hb-graph} for an illustration.
We refer to \cref{sec:proofs_hb} for the correctness, which concludes the proof of \thmref{hb-ov-hard}.
\pgfdeclarelayer{bg}    
\pgfsetlayers{bg,main}  
\begin{figure}[t]
\begin{tikzpicture}[every node/.style={scale=0.8}]

\begin{pgfonlayer}{main}
\node (start1) [vertex] {$w(z)$};
\node (start2) [vertex,right of=start1,xshift=1cm] {$\mathbf{\color{red}w(z)}$};
\node (start3) [vertex,right of=start2,xshift=1cm] {$w(z)$};

\node (lx1) [vertex,below of=start1] {$\critsec(l^{x_1})$};
\node (lx2) [vertex,below of=start2] {$\critsec(l^{x_2})$};
\node (lx3) [vertex,below of=start3] {$\critsec(l^{x_3})$};

\node (x11) [vertex,below of=lx1,xshift=-0.6cm,yshift=-0.7cm] {$\critsec(l_1)$};
\node (x13) [vertex,below of=lx1, xshift=0.5cm] {$\critsec(l_3)$};
\node (Xy13) [vertex,below of=x13] {$\critsec(l_{1,3})$};
\node (Xy23) [vertex,below of=Xy13] {$\critsec(l_{2,3})$};
\node (Xy33) [vertex,below of=Xy23] {$\critsec(l_{3,3})$};

\node (x21) [vertex,below of=lx2] {$\critsec(l_1)$};
\node (Xy11) [vertex,below of=x21] {$\critsec(l_{1,1})$};
\node (Xy21) [vertex,below of=Xy11] {$\critsec(l_{2,1})$};
\node (Xy31) [vertex,below of=Xy21] {$\critsec(l_{3,1})$};

\node (x32) [vertex,below of=lx3] {$\critsec(l_2)$};

\node (Xy12) [vertex,below of=x32] {$\critsec(l_{1,2})$};
\node (Xy22) [vertex,below of=Xy12] {$\critsec(l_{2,2})$};
\node (Xy32) [vertex,below of=Xy22] {$\critsec(l_{3,2})$};

\node (Yy11) [vertex,right of=start3, xshift=3.2cm] {$\critsec(l_{1,1})$};
\node (Yy12) [vertex,below of=Yy11] {$\critsec(l_{1,2})$};
\node (Yy13) [vertex,below of=Yy12] {$\critsec(l_{1,3})$};

\node (Yy22) [vertex,right of=Yy11, xshift=1cm] {$\critsec(l_{2,2})$};
\node (Yy23) [vertex,below of=Yy22] {$\critsec(l_{2,3})$};

\node (Yy31) [vertex,right of=Yy22, xshift=1cm] {$\critsec(l_{3,1})$};
\node (Yy32) [vertex,below of=Yy31] {$\critsec(l_{3,2})$};

\node (stop1) [vertex, below of=Yy13] {$r(z)$};
\node (stop2) [vertex, below of=Yy23] {$\mathbf{\color{red}r(z)}$};
\node (stop3) [vertex, below of=Yy32] {$r(z)$};

\draw [arrow] (start1) -- (lx1);
\draw [arrow] (start2) -- (lx2);
\draw [arrow] (start3) -- (lx3);

\draw [arrow] (lx1) -- (x11);
\draw [arrow] (lx1) -- (x13);
\draw [arrow] (x13) -- (Xy13);
\draw [arrow] (Xy13) -- (Xy23);
\draw [arrow] (Xy23) -- (Xy33);

\draw [arrow] (lx2) -- (x21);
\draw [arrow] (x21) -- (Xy11);
\draw [arrow] (Xy11) -- (Xy21);
\draw [arrow] (Xy21) -- (Xy31);

\draw [arrow] (lx3) -- (x32);
\draw [arrow] (x32) -- (Xy12);
\draw [arrow] (Xy12) -- (Xy22);
\draw [arrow] (Xy22) -- (Xy32);

\draw [arrow] (Yy11) -- (Yy12);
\draw [arrow] (Yy12) -- (Yy13);
\draw [arrow] (Yy13) -- (stop1);

\draw [arrow] (Yy22) -- (Yy23);
\draw [arrow] (Yy23) -- (stop2);

\draw [arrow] (Yy31) -- (Yy32);
\draw [arrow] (Yy32) -- (stop3);
\end{pgfonlayer}

\begin{pgfonlayer}{bg}
\draw [arrow2] (x11) -- (x21);
\draw [arrow2] (Xy12) -- (Yy12);
\draw [arrow2] (Xy31) -- (Yy31);
\draw [arrow2] (Xy23) -- (Yy23);
\draw [arrow2] (Xy13) -- (Yy13);
\draw [arrow2] (Xy11) -- (Yy11);
\draw [arrow2] (Xy22) -- (Yy22);
\end{pgfonlayer}

\node (threadsY) [textbox, right of=lx3,yshift=-3.2cm,xshift=6cm, text width=6.2cm] {\small For vectors in $Y,$ $n$ threads with $n\times d$ locks, where thread $i$ has critical section of lock $(i,k)$ if $y_i[k]=1.$ We end each thread with $\rd(z)$.};

\end{tikzpicture}
\caption{Reducing \orthv{} to finding $\HB$ races using the instance of \figref{ov-hb-example}.
For simplicity, we show the graph $\hbgraph$ instead of the trace $\tr$.
The $\HB$ race is marked in red, corresponding to the orthogonal pair $(x_2, y_2)$.
 }
\figlabel{ov-hb-graph}
\end{figure}

We now turn our attention to the problem of detecting a single $\HB$ race (i.e., not necessarily involving a read event).
We define a useful multi-connectivity problem on graphs.


\begin{problem}\problabel{multi-conn}
[\multiconn{}] Given a directed graph $G$ with $n$ nodes and $m$ edges, and $k$ pairs of nodes $(s_i,t_i),i\in[k],$ decide if there is a path in $G$ from every $s_i$ to the corresponding $t_i.$
\end{problem}

Due to \lemref{num-consecutive}, detecting whether there is an $\HB$ race in $\tr$ reduces to testing \multiconn{}
between all $O(\NumEvents)$ pairs of consecutive conflicting events in $\tr$.

\Paragraph{Short witnesses for $\HB$ races.}
We now prove \thmref{hb-decide-no-seth-better-than-3/2}.
Following \cite[Corollary~2]{CarmosinoGIMPS2016}, it suffices to show that deciding \multiconn{} can be done in $\ntime{\NumEvents^{3/2}}\cap \contime{\NumEvents^{3/2}}$.
At a first glance, the bound $\ntime{\NumEvents^{3/2}}$ may seem too optimistic,
as there are $\Theta(\NumEvents)$ paths $P_i\colon s_i\Path t_i$,
and each of them can have size $\Theta(\NumEvents)$.
Hence even just guessing these paths appears to take quadratic time.
Our proof shows that more succinct witnesses exist.

\begin{proof}[Proof of \thmref{hb-decide-no-seth-better-than-3/2}]
First consider the simpler case where $\tr$ has an $\HB$-race.
Phrased as a \multiconn{} problem on $\hbgraph$, it suffices to show that there is a pair $(s_i,t_i)$ such that $s_i$ does not reach $t_i$.
We construct a non-deterministic algorithm for this task that simply guesses the pair $(s_i,t_i)$, and verifies that there is no $s_i\Path t_i$ path. Since $\hbgraph$ is sparse, this can be easily verified in $O(\NumEvents)$ time.

Now consider the case when there is no $\HB$-race. 
Phrased as a \multiconn{} problem on $\hbgraph$, it suffices to verify that for every pair $(s_i,t_i)$, we have that $s_i$ reaches $t_i$.
We construct a non-deterministic algorithm for this task, as follows.
The algorithm operates in two phases, using a set $A$, initialized as $A=\{(s_i, t_i)\}_{i\in k}$.
\begin{compactenum}
\item In the first phase, the algorithm repeatedly guesses a node $u$ that lies on at least $\NumEvents^{1/2}$ paths $s_i\Path t_i$, for $(s_i, t_i)\in A$.
It verifies this guess via a backward and a forward traversal from $u$.
The algorithm then removes all such $(s_i, t_i)$ from $A$, and repeats.
\item In the second phase, the algorithm guesses for every remaining $(s_i, t_i)\in A$ a path $P_i\colon s_i\Path t_i$, and verifies that $P_i$ is a valid path.
\end{compactenum}
Phase 1 can be execute at most $\NumEvents^{1/2}$ iterations, while each iteration takes $O(\NumEvents)$ time since $\hbgraph$ is sparse.
Hence the total time for phase $1$ is $O(\NumEvents^{3/2})$.
Phase 2 takes $O(\NumEvents^{3/2})$ time, as every node of $\hbgraph$ appears in at most $\NumEvents^{1/2}$ paths $P_i$.
The desired result follows.
\end{proof}

\Paragraph{A super-linear lower bound for general $\HB$ races.}
Finally, we turn our attention to \thmref{hb-decide-mc-hard}.
The problem \fofee{} takes as input a first-order formula $\phi$ with quantifier structure $\forall \exists \exists$ and whose atoms are tuples,
and the task is to verify whether $\phi$ has a model on a structure of $n$ elements and $m$ relational tuples.
For simplicity, we can think of the structure as a graph $G$ of $n$ nodes and $m$ edges,
and $\phi$ a formula that characterizes the presence/absence of edges (e.g., $\phi=\forall x\exists y \exists z~e(x,y) \land \neg e(y,z)$).

The crux of the proof of \thmref{hb-decide-mc-hard} is showing the following lemma.
\begin{restatable}{lemma}{fomulticonn}\lemlabel{fo-multiconn}
\fofee{} reduces to \multiconn{} on a graph $G$ with $O(n)$ nodes in $O(n^2)$ time.
\end{restatable}

Finally, we arrive at \thmref{hb-decide-mc-hard} by constructing in $O(n^2)$ time a trace $\tr$ with $\NumEvents=\Theta(n^2)$ such that $\hbgraph$ is similar in structure to the graph $G$ of \lemref{fo-multiconn}.
In the end, detecting an $\HB$ race in $\tr$ in $O(\NumEvents^{1+\epsilon})$ time yields an algorithm
for \fofee{} in $\Theta(n^{2+\epsilon'})$ time.
We refer to \cref{sec:proofs_hb} for the details, which conclude the proof of \thmref{hb-decide-mc-hard}.


\section{Synchronization-Preserving Races}\label{sec:syncp}

In this section we discuss the dynamic detection of sync-preserving races, and prove \thmref{syncp-ov3-hard}. 

For notational convenience, we will frequently use the composite \emph{sync} events.
A $\sync(\ell)$ event represents the sequence $\acq(\ell), \rd(x_{\ell}), \wt(x_{\ell}), \rel(\ell)$.
The key idea behind sync events is as follows.
Assume that in a trace $\tr$ we have two $\sync(\ell)$ events $e_1$ and $e_2$ with $e_1\stricttrord{\tr} e_2$.
Then any correct reordering $\rho$ of $\tr$ with $e_2\in \events{\rho}$ satisfies the following.
\begin{compactenum}[(a)]
\item We have $e_1\in \events{\rho}$, as the read event of $e_2$ must read from the write event of $e_1$.
\item For every $e'_1, e'_2\in \events{\rho}$ such that $e'_1\tho{\tr}e_1$ and  $e_2\tho{\tr}e'_2$, we have
$e'_1\stricttrord{\rho} e'_2$.
\end{compactenum}
We hence use sync events to ensure certain orderings in any sync-preserving correct reordering of $\tr$ that exposes a sync-preserving data race.


\Paragraph{Intuition.}
Before we proceed with the detailed reduction, we provide a high-level description.
The input to \orthvk{3} is three sets of vectors $A_1=\{x_i\}_{i\in [n]}$, $A_2=\{y_i\}_{i\in [n]}$, and $A_3=\{z_i\}_{i\in [n]}$.
Every vector $x\in A_1$ is represented by a thread $t_x$, ending with the critical section $\acq(X), \wt(z), \rel(X)$.
Similarly, every vector $y\in A_1$ is represented by a thread $t_y$, ending with the critical section $\acq(Y), \rd(z), \rel(Y)$.
Notice that we can only have a race between the write event of a thread $t_x$ and the read event of a thread $t_y$.
The search for such a race corresponds to the search of the corresponding vectors $x\in A_1$ and $y\in A_2$ such that there is a vector $z\in A_3$ which makes the triplet $x,y,z$ orthogonal.

To establish this correspondence, we insert in $t_x$ empty critical sections on locks $l_k$, for $k\in [d]$ that represent the coordinates $k$ for which $x[k]=1$.
We use a similar encoding with locks $l'_k$ for the threads $t_y$, capturing that $y[k]=1$.
To encode the vectors in $ A_3$, we use $k$ threads $t_k$, for $k\in [d]$, such that the $i^{th}$ segment of $t_k$ encodes $z_i[k]$:~we have two interleaved critical sections on locks $l_k$ and $l'_k$ iff $z_i[k]=1$.

Finally, we use some $\sync$ events to force all threads $t_k$ be partially executed whenever we want to execute the write event of any thread $t_x$.
Hence, any correct reordering of $\tr$ that exposes a data race in $\tr$, must execute all $t_k$ at least partially.
We make all threads $t_k$ execute before all $t_x$ and $t_y$ in $\tr$.
The notion of sync-preservation ensures that if we have a correct reordering that exposes a race between two threads $t_x$ and $t_y$, then the following holds.
For every coordinate $k\in [d]$ in which $x[k]=y[k]=1$, since the corresponding threads $t_x$ and $t_y$ have critical sections on locks $l_k$ and $l'_k$, the thread $t_k$ must execute up to a point where it does not have critical sections on these locks. 
This means that we have found a vector $z$ with $z[k]=0$, and thus the triplet $x,y,z$ is orthogonal on that coordinate.

\begin{figure}
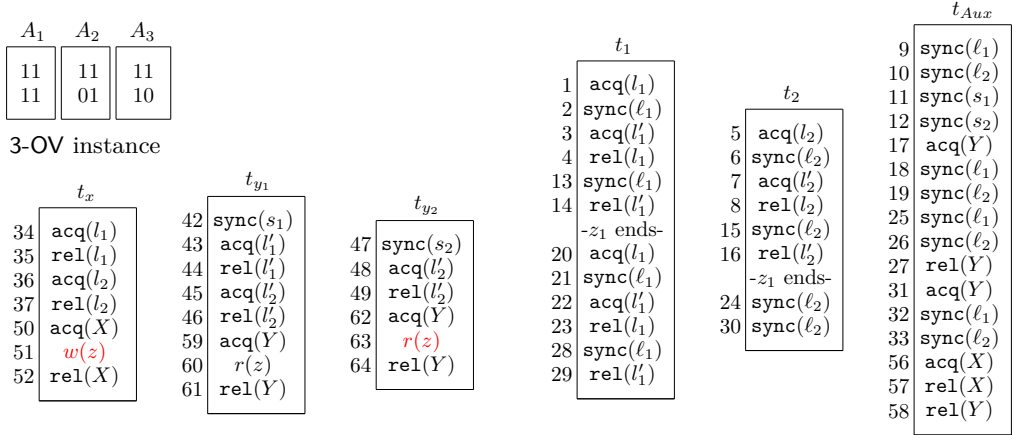

\begin{subfigure}{0.5\textwidth}
\begin{subfigure}{.5\textwidth}
\scalebox{0.8}{
\OV{
\ovX{11}
\ovX{11}
}{
\ovY{11}
\ovY{01}
}{
\ovZ{11}
\ovZ{10}
}
}
\caption*{\small \quad \orthvk{3} instance}
\end{subfigure} \\
\begin{subfigure}{0.3\textwidth}
\scalebox{0.8}{
\executionOffset{1}{x}{
	\figevoffset{33}{1}{$\acq(l_1)$}
	\figevoffset{33}{1}{$\rel(l_1)$}
	\figevoffset{33}{1}{$\acq(l_2)$}
	\figevoffset{33}{1}{$\rel(l_2)$}
	\figevoffset{45}{1}{$\acq(X)$}
    \figevoffset{45}{1}{$\color{red}w(z)$}
	\figevoffset{45}{1}{$\rel(X)$}
}
}
\end{subfigure}
\begin{subfigure}{0.3\textwidth}
\scalebox{0.8}{
\executionOffset{1}{y_1}{
	\figevoffset{41}{1}{$\sync(s_1)$}	\figevoffset{41}{1}{$\acq(l'_1)$}
	\figevoffset{41}{1}{$\rel(l'_1)$}
	\figevoffset{41}{1}{$\acq(l'_2)$}
	\figevoffset{41}{1}{$\rel(l'_2)$}
	\figevoffset{53}{1}{$\acq(Y)$}
    \figevoffset{53}{1}{$r(z)$}
	\figevoffset{53}{1}{$\rel(Y)$}
}
}
\end{subfigure}
\begin{subfigure}{0.3\textwidth}
\scalebox{0.8}{
\executionOffset{1}{y_2}{
	\figevoffset{46}{1}{$\sync(s_2)$}	
	\figevoffset{46}{1}{$\acq(l'_2)$}
	\figevoffset{46}{1}{$\rel(l'_2)$}
	\figevoffset{58}{1}{$\acq(Y)$}
    \figevoffset{58}{1}{$\color{red}r(z)$}
	\figevoffset{58}{1}{$\rel(Y)$}
}
}
\end{subfigure}
\end{subfigure}
\begin{subfigure}{0.5\textwidth}
\begin{subfigure}{0.3\textwidth}
\scalebox{0.8}{
\executionOffset{1}{1}{
	\figevoffset{0}{1}{$\acq(l_1)$}
	\figevoffset{0}{1}{$\sync(\lk_1)$}
	\figevoffset{0}{1}{$\acq(l'_1)$}
	\figevoffset{0}{1}{$\rel(l_1)$}
	\figevoffset{8}{1}{$\sync(\lk_1)$}
	\figevoffset{8}{1}{$\rel(l'_1)$}
	\figevStale{1}{-$z_1$ ends-}
	\figevoffset{12}{1}{$\acq(l_1)$}
	\figevoffset{12}{1}{$\sync(\lk_1)$}
	\figevoffset{12}{1}{$\acq(l'_1)$}
	\figevoffset{12}{1}{$\rel(l_1)$}
	\figevoffset{16}{1}{$\sync(\lk_1)$}
	\figevoffset{16}{1}{$\rel(l'_1)$}
}
}
\end{subfigure}
\begin{subfigure}{0.3\textwidth}
\scalebox{0.8}{
\executionOffset{1}{2}{
	\figevoffset{4}{1}{$\acq(l_2)$}
	\figevoffset{4}{1}{$\sync(\lk_2)$}
	\figevoffset{4}{1}{$\acq(l'_2)$}
	\figevoffset{4}{1}{$\rel(l_2)$}
	\figevoffset{10}{1}{$\sync(\lk_2)$}
	\figevoffset{10}{1}{$\rel(l'_2)$}
	\figevStale{1}{-$z_1$ ends-}
    \figevoffset{16}{1}{$\sync(\lk_2)$}
	\figevoffset{21}{1}{$\sync(\lk_2)$}
}
}
\end{subfigure}
\begin{subfigure}{0.3\textwidth}
\scalebox{0.8}{
\executionOffset{1}{Aux}{
	\figevoffset{8}{1}{$\sync(\lk_1)$}
	\figevoffset{8}{1}{$\sync(\lk_2)$}
	\figevoffset{8}{1}{$\sync(s_1)$}
	\figevoffset{8}{1}{$\sync(s_2)$}
	\figevoffset{12}{1}{$\acq(Y)$}
	\figevoffset{12}{1}{$\sync(\lk_1)$}
	\figevoffset{12}{1}{$\sync(\lk_2)$}
    \figevoffset{17}{1}{$\sync(\lk_1)$}
	\figevoffset{17}{1}{$\sync(\lk_2)$}
	\figevoffset{17}{1}{$\rel(Y)$}
	\figevoffset{20}{1}{$\acq(Y)$}
	\figevoffset{20}{1}{$\sync(\lk_1)$}
	\figevoffset{20}{1}{$\sync(\lk_2)$}
	\figevoffset{42}{1}{$\acq(X)$}
	\figevoffset{42}{1}{$\rel(X)$}
	\figevoffset{42}{1}{$\rel(Y)$}
}
}
\end{subfigure}
\end{subfigure}
\caption{
Example reduction from \orthvk{3} to $\sync$ race detection. 
The trace orders events as shown by their numbering.
We only show one thread $t_x$, as the two $x$ vectors are identical.
 }
\figlabel{ex-syncp-race}
\end{figure}

\myparagraph{Reduction}{ 
Given an  \orthvk{3} instance \ovinsthree{3} on vector sets $A_1=\{x_i\}_{i\in [n]}$, $A_2=\{y_i\}_{i\in [n]}$, and $A_3=\{z_i\}_{i\in [n]}$, we create a trace $\tr$ as follows (see \figref{ex-syncp-race}). 
We have $\NumThreads=2\cdot n + d + 1$ threads,
while all access events (not counting the sync events) are of the form $\wt(z)/\rd(z)$ in a single variable $z$.
We first describe the threads, and then how they interleave in $\tr$.

\SubParagraph{Threads.}
We introduce a thread $t_x$ for every vector $x\in A_1$ and a lock $l_k$ for every $k\in[d]$.
Each thread $t_x$ consists of two segments $t^1_x$ and $t^2_x$.
We create $t^1_x$ as follows.
For every $k\in [d]$ where $x[k]=1,$ we add an empty critical section $\acq(l_k),\rel(l_k)$  in $t^1_x$. 
We create $t^2_x$ as the sequence $\acq(X),\wt(z),\rel(X)$, where $X$ is a new lock, common for all $t^2_x$.

For the vectors in $A_2$, we introduce threads similar to those of part $A_1,$ as follows. 
We have a thread $t_y$ for every vector $y\in A_2$ and a lock $l'_k$ for every $k\in[d]$.
Each thread $t_y$ consists of two segments $t^1_y$ and $t^2_y$.
For every $k\in [d]$ where $y[k]=1,$ we add an empty critical section $\acq(l'_k),\rel(l'_k)$  in $t^1_y$. 
In contrast to the $t^1_x$, every $t^1_y$ also has an event $\sync(s_y)$ at the very beginning.
We create $t^2_y$ as the sequence $\acq(Y),\rd(z),\rel(Y)$, where $Y$ is a new lock, common for all $t^2_y$.

The construction of the threads corresponding to the vectors in $A_3$ is more involved.
We have one thread $t_k$ for every $k \in [d]$. 
Each thread has some fixed $\sync$ events, as well as critical sections corresponding to one coordinate of all $n$ vectors in $A_3$. 
In particular, we construct each $t_{k}$ as follows.
We iterate over all $z_i$, and if $z_i[k]=0$, we simply append two events $ \sync(\ell_k),\sync(\ell_k)$ to $t_k$.
On the other hand, if $z_i[k]=1$, we interleave these sync events with two critical sections, by appending the sequence
$
\acq(l_k), \sync(\ell_k), \acq(l'_k), \rel(l_k), \sync(\ell_k), \rel(l'_k)
$.

Lastly, we have a single auxiliary trace $t$ that consists of three parts $t^1$, $t^2$ and $t^3$, where
\begin{align*}
t^1&= \sync(\ell_1),\dots, \sync(\ell_k), \sync(s_{y_1}),\dots \sync(s_{y_n})\\
t^2&=\left(\acq(Y), \sync(\ell_1),\dots, \sync(\ell_k),  \sync(\ell_1),\dots, \sync(\ell_k), \rel(Y)\right)^{n-1}\\
t^3&=\acq(Y), \sync(\ell_1),\dots, \sync(\ell_k), \acq(X), \rel(X), \rel(Y)
\end{align*}

\SubParagraph{Concurrent trace.}
We are now ready to describe the interleaving of the above threads in order to obtain the concurrent trace $\tr$.
\begin{compactenum}
\item We execute the auxiliary trace $t$ and all traces $t_k$, for $k\in [d]$ (i.e., the threads corresponding to the vectors of $A_3$) arbitrarily, as long as for every $k\in[d]$, every sequence of $\sync(\ell_k)$ events
\begin{enumerate*}[(a)]
\item starts with the $\sync(\ell_k)$ event of $t_k$ and proceeds with the $\sync(\ell_k)$ event of $t$,
\item strictly alternates in every two $\sync(\ell_k)$ events between $t$ and $t_k$, and
\item ends with the last $\sync(\ell_k)$ event of $t_k$.
\end{enumerate*}
\item We execute all $t^1_x$ and $t^1_y$ (i.e., the first parts of all threads that correspond to the vectors in $A_1$ and $A_2$) arbitrarily, but after all traces $t_k$, for $k\in [d]$.
\item We execute all $t^2_x$ (i.e., the second parts of all traces that correspond to the vectors in $A_1$) arbitrarily, but before the segment $\acq(X), \rel(X), \rel(Y)$ of $t$.
\item We execute all $t^2_y$ (i.e., the second parts of all traces that correspond to the vectors in $A_2$) arbitrarily, but after the segment $\acq(X), \rel(X), \rel(Y)$ of $t$.
\end{compactenum}
}

We refer to \cref{sec:proofs_syncp} for the correctness of the reduction and thus the proof of \thmref{syncp-ov3-hard}.



\section{Violations of the Locking Discipline}\label{sec:locking_discipline}


\subsection{Lock-Cover Races}

We start with a simple reduction from \orthv{} to detecting lock-cover races.
Given a \orthv{} instance \ovinstwo{} on two vector sets $A_1, A_2$, we create a trace $\tr$ as follows. 
We have a single variable $x$ and two threads $t_1, t_2$.
We associate with each vector of the set $A_i$ a write access event $e = \ev{t_i, \wt(x)}$.
Moreover, each such event holds up to $d$ locks, so that $e$ holds the $k^{th}$ lock iff $k^{th}$ coordinate of the vector corresponding to the event is $1$.
The trace $\tr$ is formed by ordering the sequence of events corresponding to vectors of $A_1$ of \orthv{} first, in a fixed arbitrary order, followed by the sequence of events corresponding to $A_2$, again in arbitrary order.
We refer to \cref{sec:proofs_locking_discipline} for the correctness, which concludes the proof of \thmref{lock-cover-quadratic-lower-bound}.

\subsection{Lock-Set Races}

We now turn our attention to lock-set races.
We first prove \thmref{lock-set-single-variable-linear}, i.e., that determining whether 
a trace $\tr$ has a lock-set race on a specific variable $x$ can be performed in linear time.

\Paragraph{A linear-time algorithm per variable.}
Verifying that there are two conflicting events on $x$ is straightforward by a single pass of $\tr$.
The more involved part is in computing the lock-set of $x$, i.e., the set $\bigcap_{e \in \accesses{\tr}(x)} \lheld{\tr}(e)$, in linear time.
Indeed, each intersection alone requires $\Theta(\NumLocks)$ time, resulting to $\Theta(\NumEvents\cdot \NumLocks)$ time overall.

Here we show that a somewhat more involved algorithm achieves the task.
The algorithm performs a single pass of $\tr$, while maintaining three simple sets $A$, $B$, and $C$.
While processing an event $e$, the sets are updated to maintain the invariant
\begin{align}\label{eq:lockset_invariant}
A=\lheld{\tr}(e)\quad\;
B=\locks{\tr}\cap \bigcap_{e' \in \accesses{\tr}(x), e'\trord{\tr} e} \lheld{\tr}(e') \quad\;
C=\ov{A}\cap B
\end{align}

The sets are initialized as $A=\emptyset$, $B=C=\locks{\tr}$.
Then the algorithm performs a pass over $\tr$ and processes each event $e$ according to the description of \algoref{compute-lockset-variable}.


{
\footnotesize
\begin{algorithm*}[h]
\begin{multicols}{4}
\myhandler{\acqhandler{$t$, $\lk$}}{
	$A\gets A\cup \set{\lk}$ \;
	\If{$\lk \in B$}{
		$C \gets C \setminus \set{\lk}$
	}
}

\myhandler{\relhandler{$t$, $\lk$}}{
	$A\gets A\setminus \set{\lk}$ \;
	\If{$\lk \in B$}{
		$C \gets C \cup \set{\lk}$
	}
}

\myhandler{\rdhandler{$t$, $y$}}{
	\If{$x = y$}{
	$B \gets B \setminus C$ \;
	$C \gets \emptyset$
	}
}

\myhandler{\wthandler{$t$, $y$}}{
	\If{$x = y$}{
	$B \gets B \setminus C$ \;
	$C \gets \emptyset$
	}
}

\end{multicols}
\normalsize
\caption{\textit{Computing lock-set of variable $x$}}
\algolabel{compute-lockset-variable}
\SetAlgoInsideSkip{medskip}
\end{algorithm*}
\normalsize
}

The correctness of \algoref{compute-lockset-variable} follows by proving the invariant in \cref{eq:lockset_invariant}.
We refer to \cref{sec:proofs_locking_discipline} for the details, which concludes the proof of \thmref{lock-set-single-variable-linear}.

\Paragraph{Short witnesses for lock-set races.}
Besides the advantage of a faster algorithm, \thmref{lock-set-single-variable-linear} implies that 
lock-set races have short witnesses that can be verified in linear time.
This allows us to prove that detecting a lock-set race is in $\ntime{\NumEvents}\cap \contime{\NumEvents}$, and we can thus use \cite[Corollary~2]{CarmosinoGIMPS2016} to prove \thmref{lock-set-nseth}.

\begin{proof}[Proof of \thmref{lock-set-nseth}]
First we argue that the problem is in $\ntime{\NumEvents}$.
Indeed, the certificate for the existence of a lock-set race is simply the variable $x$ on which there is a lock-set race.
By \thmref{lock-set-single-variable-linear}, verifying that we indeed have a lock-set race on $x$ takes $O(\NumEvents)$ time.
 
Now we argue that the problem is in $\contime{\NumEvents}$, by giving a certificate to verify in linear time that $\tr$ does not have a race of the required form. 
The certificate has size $O(|\vars{\tr}|)$, and specifies for every variable, either the lock that is held by all access events of the variable, or a claim that there exist no two conflicting events on that variable.
The certificate can be easily verified by one pass over $\tr$. 
\end{proof}

\Paragraph{Lock-set races are Hitting-Set hard.}
Finally we prove \thmref{lock-set-quadratic-lower-bound}, i.e., that determining a single lock-set race is \hs{}-hard, and thus also carries a conditional quadratic lower bound.
We establish a fine-grained reduction from \hs{}. 
Given a \hs{} instance \hsins{} on two vector sets $X, Y$, we create a trace $\tr$ using $d+1$ threads $\{t_j\}_{j\in \{0\}\cup [d]}$, $n$ locks $\{\ell_i\}_{i\in[n]}$, and $n$ variables $\{z_i\}_{k\in [n]}$.
Thread $t_0$ that executes
$
\acq(\ell_1),\dots, \acq(\ell_n), \wt(z_1),\dots \wt(z_n), \rel(\ell_n), \dots \rel(\ell_1)
$.
Each of the threads $t_j$, for $j\in[d]$, has a single nested critical section consisting of the locks $\ell_i\in [n]$ such that the $i^{th}$ vector of $Y$ has its $j^{th}$ coordinate $0$, i.e, $y_i[j]=0$.
The events in the critical section are all write events of all variables $z_k\in [n]$ with $x_k[j]=1$.
The trace orders all events of each thread $t_d$ consecutively, and all the events overall in increasing order of $d$.
See \cref{fig:hs_to_lock_set} for an illustration.
We refer to \cref{sec:proofs_locking_discipline} for the correctness, which concludes the proof of \thmref{lock-set-quadratic-lower-bound}.

\begin{figure}
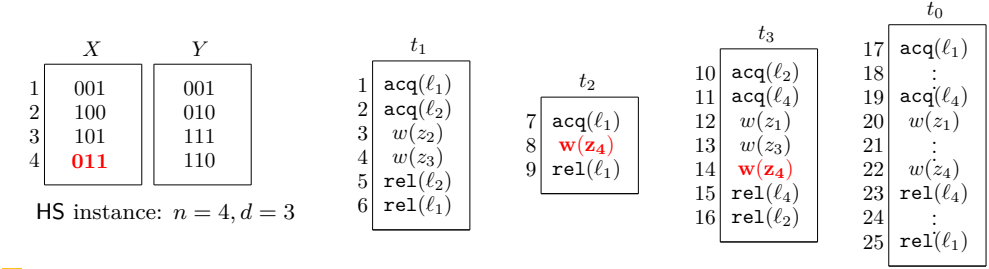

\begin{subfigure}{.3\textwidth}
\scalebox{0.8}{
\HS{
	\hsevX{001}
	\hsevX{100}
    \hsevX{101}
    \hsevX{\textbf{\textcolor{red}{011}}}
}{
	\hsevY{001}
	\hsevY{010}
    \hsevY{111}
    \hsevY{110}
}
}
\caption*{\hspace{1.5em}\hs{} instance: $n=4,d=3$}
\end{subfigure}
\begin{subfigure}{.15\textwidth}
\scalebox{0.8}{
\executionOffset{1}{1}{
	\figev{1}{$\acq(\lk_1)$}
	\figev{1}{$\acq(\lk_2)$}
    \figev{1}{$w(z_2)$}
    \figev{1}{$w(z_3)$}
	\figev{1}{$\rel(\lk_2)$}
	\figev{1}{$\rel(\lk_1)$}
}
}
\end{subfigure}
\begin{subfigure}{.15\textwidth}
\scalebox{0.8}{
\executionOffset{1}{2}{
    \figevoffset{6}{1}{$\acq(\lk_1)$}
    \figevoffset{6}{1}{$\mathbf{\color{red} w(z_4)}$}
    \figevoffset{6}{1}{$\rel(\lk_1)$}
}
}
\end{subfigure}
\begin{subfigure}{.15\textwidth}
\scalebox{0.8}{
\executionOffset{1}{3}{
    \figevoffset{9}{1}{$\acq(\lk_2)$}      \figevoffset{9}{1}{$\acq(\lk_4)$}
    \figevoffset{9}{1}{$w(z_1)$}
    \figevoffset{9}{1}{$w(z_3)$}
    \figevoffset{9}{1}{$\mathbf{\color{red}w(z_4)}$}
    \figevoffset{9}{1}{$\rel(\lk_4)$}
    \figevoffset{9}{1}{$\rel(\lk_2)$}}
}
\end{subfigure}
\begin{subfigure}{.15\textwidth}
\scalebox{0.8}{
\executionOffset{1}{0}{
	\figevoffset{16}{1}{$\acq(\lk_1)$}
	\figevoffset{16}{1}{$\vdots$}
	\figevoffset{16}{1}{$\acq(\lk_4)$}
    \figevoffset{16}{1}{$w(z_1)$}
    \figevoffset{16}{1}{$\vdots$}
    \figevoffset{16}{1}{$w(z_4)$}
	\figevoffset{16}{1}{$\rel(\lk_4)$}
	\figevoffset{16}{1}{$\vdots$}
	\figevoffset{16}{1}{$\rel(\lk_1)$}
}
}
\end{subfigure}
\caption{Reducing \hs{} to detecting a lock-set race on trace $\tr$ with $d$ threads. 
Thread $t_k$ uses lock $l_i$ if $y_i[k]=0$, and $w(z_j)$ if $x_j[k]=1$. 
Vector $x_4$ hits all vectors in $Y$, implying a lock-set race on $z_4$.
}
\label{fig:hs_to_lock_set}
\end{figure}

\section{Conclusion}
In this work we have taken a fine-grained view of the complexity of popular notions of dynamic data races.
We have established a range of lower bounds on the complexity of detecting $\HB$ races, sync-preserving races, as well as races based on the locking discipline (lock-cover/lock-set races).
Moreover, we have characterized cases where lower bounds based on SETH are not possible under NSETH.
Finally, we have proven new upper bounds for detecting $\HB$ and lock-set races.
To our knowledge, this is the first work that characterizes the complexity of well-established dynamic race-detection techniques, allowing for a rigorous characterization of their trade-offs between expressiveness and running time.

\newpage
\bibliography{bibliography}

\newpage
\appendix
\section{Fine-Grained Complexity and Popular Hypotheses}\label{sec:fine_grained}

In this section we present notions of fine-grained complexity theory that are relevant to our work.
We refer to the survey \cite{Williams18} for a detailed exposition on the topic.

This theory relates the computational complexity of problems under the following, more refined, notion of reduction than the standard ones used in traditional complexity theory. Informally, the definition says that if there is an algorithm for some problem B faster than its assumed lower bound, then such a reduction from some problem A to B gives an algorithm for A thatis faster than its conjectured lower bound. 

\Paragraph{Fine-grained Reductions.}
Assume that A and B are computational problems and $a(n)$ and $b(n)$ are their conjectured running time lower bounds, respectively. Then we say A $(a,b)$-reduces to B, denoted by
A \fgr{a}{b} B, if for every $\epsilon>0,$ there exists $\delta > 0$, and an algorithm R for A that runs in time $a(n)^{(1-\delta)}$ on
inputs of length $n$, making $q$ calls to an oracle for B with query lengths $n_1,\dotsc, n_q$, where,
$$\sum_{1}^{q}(b(n))^{(1-\epsilon)}\le (a(n))^{(1-\delta)}.$$

Problems that can be reduced to each other such that the lower bounds for each problem are the same in both reductions, i.e., A\fgr{a}{b}B and B\fgr{b}{a}A, are intuitively thought to have the same underlying `reason' for hardness, and are said to be fine-grained equivalent. 

A reduction A\fgr{a}{b}B would be interesting for B if $a(n)$ was a proven or well-believed conjectured lower bound on A, thus implying a believable lower bound on B. One such well-believed conjecture in complexity theory is \seth{}~\cite{ImpagliazzoP01} for the classic CNF-SAT problem, originally defined for deterministic algorithms, but now widely believed for randomized algorithms as well.


\begin{hypothesis}
[Strong Exponential Time Hypothesis (\seth)]\deflabel{seth}
For every $\epsilon>0$ there exists an integer $k\ge 3$
such that CNF-SAT on formulas with clause size at most $k$ and $n$ variables cannot be solved in $O(2^{(1-\epsilon)n})$ time even by a randomized algorithm.
\end{hypothesis}

\seth{} implies a lower bound conjecture, denoted by \ovc{}, 
on the Orthogonal Vectors problem \orthv{}, as shown by a 
reduction from CNF-SAT to $\orthvk{k}$ \cite{Williams05}. Thus, a conditional lower bound under \ovc{} implies one under \seth{} as well, leading to numerous conditional lower bound results under \ovc{} [See \cite{Williams18} for a detailed literature review]. This paper will also prove such results on several data race detection problems, hence we now state $\orthvk{k}$ and \ovc{} formally. 

An instance of $\orthvk{k}$ is an integer $d=\omega(\log n)$ and $k$ 
sets $A_i\subseteq \{0,1\}^d,\ i\in [n]$ such that $|A_i|=n,$ and denoted by $\ovinstwo{}$. 
\begin{problem}[Orthogonal Vectors ($\orthvk{k}$)]\problabel{ov-definition}
Given an instance $\ovinsthree{k}$, the \orthvk{k} problem is to decide if there are $k$ vectors $a_i\in A_i$ for all $i\in [n]$ such that the sum of their point wise product is zero, i.e., $\sum_{j=1}^d \prod_{i=1}^k a_i[j]=0.$
\end{problem}
For ease of exposition, we denote $\ovinsthree{2}$ and $\orthvk{2}$ by \ovinstwo{} and \orthv{} respectively.

\begin{hypothesis}
[Orthogonal Vectors Hypothesis (\ovc)]\deflabel{ovc}
No randomized algorithm can solve $\orthvk{k}$ for an instance $\ovinsthree{k}$ in time $O(n^{(k-\epsilon)}\cdot \poly(d))$ for any constant $\epsilon>0$.
\end{hypothesis}

There is an impossibility result from \cite{CarmosinoGIMPS2016} that proves that a reduction under \seth{}, and hence under \ovc{}, is not possible unless the following \nseth{} conjecture is false. 

\begin{hypothesis}
[Non-deterministic \seth{} (\nseth)]\deflabel{nseth}
 For every $\epsilon>0$, there exists a $k$ so that
k-TAUT is not in $\ntime{2^{n(1-\epsilon)}}$, where k-TAUT is the language of all k-DNF formulas which are tautologies.
\end{hypothesis}

The impossibility result \cite[Corollary~2]{CarmosinoGIMPS2016} is as follows.

\begin{theorem}\thmlabel{better-than-t-unlikely-under-seth}
If \nseth{} holds and a problem C $\in \ntime{T_C}\cap \contime{T_C},$ then for any problem B that is \seth{}-hard under deterministic reductions with time $T_B,$ and $\gamma>0,$ we cannot have a fine-grained reduction B \fgr{T_B}{c} C where $c=T_C^{(1+\gamma)}.$
\end{theorem}

We show some of our problems satisfy the conditions of \thmref{better-than-t-unlikely-under-seth}, and hence show lower bounds for these conditioned on one of two other hypotheses called \hsc{} and \foh{}, described below. 

An instance of the hitting set problem, denoted by \hs{}, is an integer $d=\omega(\log n)$ and sets $X,Y\subseteq \{0,1\}^d,\ i\in [n]$ such that $|X|=|Y|=n,$ and denoted by \hsins{}. 
\begin{problem}[Hitting Sets (\hs)]\problabel{hs-definition}
Given an instance \hsins{}, the \hs{} problem is to decide if there is a vector $x\in X$ such that for all $y\in Y$ we have $x\cdot y\ne 0,$ or informally, some vector in $X$ hits all vectors in $Y.$
\end{problem}

\begin{hypothesis}
[Hitting Sets Hypothesis (\hsc)]\deflabel{hsc}
No randomized algorithm can solve \hs{} for an instance \hsins{} in time $O(n^{(2-\epsilon)}\cdot \poly(d))$ for any constant $\epsilon>0.$
\end{hypothesis}

\hsc{} implies \ovc{}, but the reverse direction is not known. 

Finally we consider a subclass of first order formula over structures of size $n$ and with $m$ relational tuples~\cite{Gao2018}.
\begin{problem}[\fofee{}]\problabel{first-order-forall-exist-exist}
Decide if a given a first-order formula quantified by $\forall\exists\exists$ has a model on a structure of size $n$ with $m$ relational tuples.
\end{problem}

It is known that \fofee{} can be solved in $O(m^{3/2})$ time using ideas from triangle detection algorithms~\cite{Gao2018}. 
For dense structures ($m=\Theta(n^2))$, this yields the bound $O(n^3)$.
Although sub-cubic algorithms might be possible, achieving a truly quadratic bound seems unlikely or at least highly non-trivial.


\section{Proofs of \cref{sec:hb}}\label{sec:proofs_hb}


\subsection{Proofs from \secref{hb-upper-bound}}
\applabel{proofs_hb_ub}

\hbisomorphiclockstamp*
\begin{proof}
($\Rightarrow$)  Let $e_1 \hb{\tr} e_2$.
Using the definition of $\hb{\tr}$, 
there must be a sequence of events $f_1,f_2 \ldots f_k$ with 
$k > 1$, $f_1 = e_1$, $f_k = e_2$, and for every $1 \leq i < k$,
$f_i \trord{\tr} f_{i+1}$ and
either $f_i \tho{\tr} f_{i+1}$ or there is a lock $\lk$, such that
$f_i \in \releases{\tr}(\lk)$ and $f_{i+1} \in \acquires{\tr}(\lk)$.
Let $j$ be the smallest index $i$ such that $\ThreadOf{f_i} \neq \ThreadOf{f_{i+1}}$;
such an index exists as $\ThreadOf{e_1} \neq \ThreadOf{e_2}$.
Observe that there must be a lock $\lk$ for which
$\OpOf{f_j} = \rel(\lk)$ and $\OpOf{f_{j+1}} = \acq(\lk)$.
Observe that $pos_\tr(f_j) < pos_\tr(f_{j+1})$,
$\rells{\tr}_{e_1}(\lk) \leq pos_\tr(f_j)$
and $pos_\tr(f_{j+1}) \leq \acqls{\tr}_{e_2}$, giving us
$\rells{\tr}_{e_1}(\lk) < \acqls{\tr}_{e_2}(\lk)$.

($\Leftarrow$) Let $\lk$ be a lock such that $\rells{\tr}_{e_1}(\lk) < \acqls{\tr}_{e_2}(\lk)$.
Then, there is a release event $f$ and an acquire event $g$ on lock $\lk$
such that $pos_\tr(f) < pos_\tr(g)$,$e_1 \hb{\tr} f$ and $g \hb{\tr} e_2$.
This means $f \hb{\tr} g$ and thus $e_1 \hb{\tr} e_2$.
\end{proof}

For the sake of completeness, we present the computation of release lockstamps.
As with \algoref{assign-acquire-ls}, we maintain the following variables.
For each thread $t$ and lock $\lk$, we
will maintain variables $\CTC_t$ and $\CTL_\lk$ that take values from 
the space of all lockstamps.
We also additionally maintain an integer variable $\idxlk_\lk$
for each lock $\lk$ that stores the index (or relative position) of 
the earliest (according to the trace order $\trord{\tr}$) 
release event of lock $\lk$ in the trace.
Initially, we set each $\CTC_t$ and $\CTL_m$ to $\lambda \lk \cdot \infty$,
for each thread $t$ and lock $m$.
Further, for each lock $m$, we set $\idxlk_m$ to $n_m+1$,
where $n_m$ is the number of release events of $m$ in the trace;
this can be obtained in a linear scan 
(or  by reading the value of $\idxlk_m$ at the end of a run of \algoref{assign-acquire-ls}).
We traverse the events according to the total trace order
and perform updates to the data structures as described in \algoref{assign-release-ls},
by invoking the appropriate \emph{handler}
based on the thread and operation of the event $e = \ev{t, op}$ being visited.
At the end of each handler, we assign the lockstamp $\rells{\tr}_e$ to the event $e$.

{
\footnotesize
\begin{algorithm*}[h]
\begin{multicols}{3}

\myhandler{\acqhandler{$t$, $\lk$}}{
	$\CTL_\lk\gets \CTC_t$ \;
	$\rells{\tr}_e \gets \CTC_t$	
}

\myhandler{\relhandler{$t$, $\lk$}}{
	$\idxlk_\lk\gets \idxlk_\lk - 1$ \;
	$\CTC_t\gets \CTC_t[\lk \mapsto \idxlk_\lk] \mn\CTL_\lk$ \;
	$\rells{\tr}_e \gets \CTC_t$
}

\myhandler{\rdhandler{$t$, $x$}}{
	$\rells{\tr}_e \gets \CTC_t$
}

\BlankLine

\myhandler{\wthandler{$t$, $x$}}{
	$\rells{\tr}_e \gets \CTC_t$
}

\end{multicols}
\normalsize
\caption{\textit{Assigning release lockstamps to events in the trace}}
\algolabel{assign-release-ls}
\SetAlgoInsideSkip{medskip}
\end{algorithm*}
\normalsize
}

Let us now state the correctness of \algoref{assign-acquire-ls}
and \algoref{assign-release-ls}.
\begin{lemma}
\lemlabel{lockstamp-assigment-correct}
On input trace $\tr$,
\algoref{assign-acquire-ls} and \algoref{assign-release-ls}
correctly compute the lockstamps $\acqls{\tr}_e$ and $\rells{\tr}_e$ respectively
for each event $e \in \events{\tr}$.
\end{lemma}

\begin{proof}[Proof Sketch]
We focus on the correctness proof of \algoref{assign-acquire-ls};
the proof for \algoref{assign-release-ls} is similar.
The proof relies on the invariant maintained by  \algoref{assign-acquire-ls}
the variables $\CTC_t$, $\CTL_\lk$ and $\idxlk_\lk$ for each thread $t$ and lock $\lk$,
which we state next.
Let $\pi$ be the prefix of the trace processed at any point in the algorithm.
Let $C^\pi_t$, $L^\pi_\lk$ and $p^\pi_\lk$ be the values of the variables
$\CTC_t$, $\CTL_\lk$ and $\idxlk_\lk$ after processing the prefix $\pi$.
Then, the following invariants are true:
\begin{itemize}
	\item $C^\pi_t = \acqls{\pi}_{e^\pi_t} = \acqls{\tr}_{e^\pi_t}$, where $e^\pi_t$ is the last event in $\pi$ performed by thread $t$
	\item $L^\pi_\lk = \acqls{\pi}_{e^\pi_\lk} = \acqls{\tr}_{e^\pi_\lk}$, where $e^\pi_\lk$ is the last acquire event on lock $\lk$ in $\pi$.
	\item $p^\pi_\lk = pos^\pi_\lk(e^\pi_\lk)$, where $e^\pi_\lk$ is the last acquire event on lock $\lk$ in $\pi$.
\end{itemize}
These invariants can be proved using a straightforward induction, each time noting the definition of $\hb{\tr}$.
\end{proof}

\begin{lemma}
\lemlabel{time-assign-lockstamps}
For a trace with $\NumEvents$ events and $\NumLocks$ locks, \algoref{assign-acquire-ls}
and \algoref{assign-release-ls} both take $O(\NumThreads \cdot \NumLocks)$ time.
\end{lemma}
\begin{proof}
We focus on \algoref{assign-acquire-ls};
the analysis for \algoref{assign-release-ls} is similar.
At each acquire event, the algorithm spends $O(1)$ time 
for updating $\idxlk_\lk$,
$O(\NumLocks)$ time for doing the $\mx$ operation,
and $O(\NumLocks)$ time for the copy operation (`$\acqls{\tr}_e \gets \CTC_t$').
For a release event, we spend $O(\NumLocks)$ for the two copy operations.
At read and write events, we spend $O(\NumLocks)$ for copy operations.
This gives a total time of $O(\NumEvents \cdot \NumLocks)$.
\end{proof}

\numconsecutive*
\begin{proof}
We first prove that  if there is a an $\HB$-race in $\tr$, 
then there is a pair of consecutive conflicting events that is in $\HB$-race.
Consider the first $\HB$-race, i.e., 
an $\HB$-race $(e_1,e_2)$ such that for every other $\HB$-race $(e'_1, e'_2)$,
either $e_2 \trord{\tr} e'_2$ or $e_2 = e'_2$ and $e'_1 \trord{\tr} e_1$.
We remark that such a race $(e_1, e_2)$ exists if $\tr$ has any $\HB$-race.
We now show that $(e_1, e_2)$ are a consecutive conflicting pair (on variable $x$).
Assume on the contrary that there is an event $f \in \writes{\tr}(x)$
such that $e_1 \stricttrord{\tr} f \stricttrord{\tr} e_2$.
If either $(e_1, f)$ or $(f, e_2)$ is an $\HB$-race, then this contradicts 
our assumption that $(e_1, e_2)$ is the first $\HB$-race in $\tr$.
Thus, $e_1 \hb{\tr} f$ and $f \hb{\tr} e_2$, which gives $e_1 \hb{\tr} e_2$, another contradiction.

We now turn our attention to the number of consecutive conflicting events in $\tr$.
For every read or write event $e_2$, there is at most
one write event $e_1$ such that $(e_1, e_2)$ is a consecutive conflicting pair
(namely the latest conflicting write event before $e_2$)
Further, for every read event $e_1$, there is at most one
write event $e_2$ such that $(e_1, e_2)$ is a consecutive conflicting pair
(namely the earliest conflicting write event after $e_1$).
This gives at most $2\NumEvents$ consecutive conflicting pairs of events.
\end{proof}

Let us now state the correctness of \algoref{check-hb-race}.
\begin{lemma}
\lemlabel{lockstamp-checking-correct}
For a trace $\tr$, \algoref{check-hb-race} reports a race iff $\tr$ has an $\HB$-race.
\end{lemma}

\begin{proof}[Proof Sketch.]
The proof relies on the following straightforward invariants;
we skip their proofs as they are straightforward.
In the following, $e^\pi_x$ is the last event with $\OpOf{e^\pi_x} = \wt(x)$
in a trace $\pi$.
\begin{itemize}
	\item After processing the prefix $\pi$ of $\tr$, $t^w_x = \ThreadOf{e^\pi_x}$ and $\VCW_x = e^\pi_x$.
	\item After processing the prefix $\pi$ of $\tr$, 
	the set $\setrdls_x$ is $\setpred{(t, L)}{\exists e \in \reads{\pi}(x), e^\pi_x \trord{\pi} e, \ThreadOf{e} = t, \rells{\tr}_e = L}$.
\end{itemize}
The rest of the proof follows from \lemref{time-assign-lockstamps} and \lemref{lockstamp-assigment-correct}.
\end{proof}

Let us now characterize the time complexity of \algoref{check-hb-race}.
\begin{lemma}
\lemlabel{lockstamp-checking-complexity}
On an input trace with $\NumEvents$ events and $\NumLocks$ locks,
\algoref{check-hb-race} runs in time $O(\NumEvents \cdot \NumLocks)$.
\end{lemma}
\begin{proof}[Proof Sketch]
Each pair $(t, L)$ of thread identifier and lockstamp is added 
atmost once in some set $\setrdls_x$ (for some $x$).
Also, each such pair is also compared against another timestamp atmost once.
Each comparison of timestamps take $O(\NumLocks)$ time.
This gives a total time of $O(\NumEvents \cdot \NumLocks)$.
\end{proof}

\hbalgo*
\begin{proof}
We focus on proving that there is an $O(\NumEvents \cdot \NumLocks)$ time algorithm,
as the standard vector-clock algorithm~\cite{djit1999} 
for checking for an $\HB$-race runs in $O(\NumEvents \cdot\NumThreads)$ time.
Our algorithm's correctness is stated in \lemref{lockstamp-checking-correct}
and its total running time is $O(\NumEvents \cdot \NumLocks)$
(\lemref{lockstamp-checking-complexity} and \lemref{time-assign-lockstamps}).
\end{proof}

\subsection{Proofs from \secref{hb-lower-bound}}
\applabel{proofs_hb_lb}

\hbovhard*
\begin{proof}
Consider a pair of events $\wt(z)$ from the $d$ threads $t(x,i),i\in[d],$ and $\rd(z)\in t_y$ for some $x,i,y$.
We have $\wt(z)\hb{\tr} \rd(z)$ iff there is some path from $\wt(z)$ to $\rd(z)$  in $\hbgraph$.
As $\wt(z)$ and $\rd(z)$ are in different threads, such a path can only be through lock events in a sequence of threads such that the first and last threads are $t(x,i)$ for some $i\in[d]$ and $t_y$, and every consecutive pair of threads in the sequence holds a common lock.  
Now all the locks in $t_y$ are $l(y,i)$ for all $i$ where $y[i]=1$.
Consider the lock corresponding to any $i\in [d]$. 
The only thread $t(x',i)$ that also holds this lock corresponds to the last $x'$ such that $x'[i]=1$. 
The only other lock held by $t(x',i)$ is $l_i.$ If $\wt(z)$ is in $t(x',i),$ we are done. 
Otherwise the only common lock between these threads $t(x',i)$ and those of $\wt(z)$ can be one of the $l_i$. 
The threads of $\wt(z)$ contain all $l_i$ where $x[i]=1.$ Hence, for there to be a common lock between these threads, there must be at least one $i$ such that $x'[i]=1$ and $x[i]=1$. 
As this thread also has the lock $l(y,i),$ $y[i]$ is also $1$. 

Thus, there is a path from $\wt(z)$ to $\rd(z)$ if and only if there is at least one $i\in[d]$ such that $x[i]=y[i]=1,$ hence $x$ and $y$ are not orthogonal. A pair of orthogonal vectors of \orthv{} thus corresponds to a write-read $\HB$-race in the reduced trace.

Finally we turn our attention to the complexity.
In time $O(n\cdot d),$ we have reduced an \orthv{} instance to determining whether there is a write-read $\HB$ race in a trace of $\NumEvents=O(nd)$ events.
If there was a sub-quadratic i.e. $O((n\cdot d)^{(2-\epsilon)})=n^{(2-\epsilon)}\cdot \poly(d)$ algorithm for detecting a write-read $\HB$ race, then this would also solve \orthv{} in $n^{(2-\epsilon)}\cdot \poly(d)$ time, refuting the \orthv{} hypothesis.
\end{proof}

\fomulticonn*
\begin{proof}
For intuition, assume the first order property is on an undirected graph with $n$ variables and $m$ edges. Let the property be specified in quantified $3$-DNF form with a constant number of predicates, i.e., $\phi=\forall x\exists y\exists z\  (\psi_1\vee \psi_2\vee\dotsc \psi_k),$ where $x,y,z$ represent nodes of the graph, and each $\psi_i$ is a conjunction of $3$ variables representing edges of the graph, for example $e(x,y)\wedge \neg e(y,z) \wedge e(x,z)$. The property is then true if and only if some predicate is satisfied, which is true if all of its variables are satisfied ($e(x,y)$ is satisfied when edge $(x,y)$ is in the graph). Denote the graph on which $\phi$ is defined by $H(I,J),$ where $I$ and $J$ are respectively the sets of nodes and edges of $H.$ 

The instance of \multiconn{} is constructed given $H$ and $\phi$ as follows. Construct a $(2k+2)$-partite graph $G(V,E)$ by first creating $2k+2$ copies of $I.$ Denote these copies by $S,Y_i,Z_i,T,\ i\in[k],$ and the copy of each node $x\in I$ in any part, say $S,$ by $x(S).$ $\psi_i=(e_1\wedge e_2 \wedge e_3)$ is encoded by connecting the sets $(S,Y_i)$ to represent $e_1$, $(Y_i,Z_i)$ for $e_2$ and $(Z_i,T)$ for $e_3$ as follows. If $e_i$ is of the form $e(x,y)$ (and not its negation), then draw a copy of $H$ between its corresponding sets, say $S$ and $Y_i$ without loss of generality. That is, for every $x,y,$ $(x,y)\in J \Leftrightarrow (x(S),y(Y_i))\in E$. If on the other hand $e_i$ is of the form $\neg e(x,y)$ then connect a copy of the complement of $H$, i.e., $(x,y)\notin J\Leftrightarrow (x(S),y(Y_i))\in E.$

Finally define $|I|$ pairs $(x(S),x(T))$ as the $(s,t)$ pairs for \multiconn{}. 

We now prove this reduction is correct. First, assume $\phi$ is true. Then for every node $x,$ there exist nodes $y,z$ such that some predicate is true. If $\psi_i$ is the predicate that is satisfied for some node $u,$ then there is a path between $u(S)$ and $u(T)$ through the parts $S,Y_i,Z_i$ and $T$ as follows. As the first variable is satisfied, then if it is $e(x,y),$ then $(x,y)\in J,$ and $x(S)$ is connected to $y(Y_i),$ and if it is $\neg e(x,y),$ then $(x,y)\notin J$ and again $x(S)$ is connected to $y(Y_i).$ Similarly, $y(Y_i)$ is connected to $z(Z_i),$ and $z(Z_i)$ to $x(T).$ These edges form a $3$ length path between $x(S)$ and $x(T)$. 

Now consider the reverse case, and assume the \multiconn{} problem is true, that is , there is a path between every $(x(S),x(T))$ pair. Note that the construction of edges in $G$ is such that any path from $x(S)$ to $x(T)$ has to be a $3$ length path, connecting the copy of $x$ in $S$ to its copy in some $Y_i,$ from this $Y_i$ to its corresponding $Z_i,$ and from $Z_i$ to $T.$ Also, this path exists only if all variables of the corresponding $\psi_i$ are true. Hence, as there is a path between every pair $(x(S),x(T)),$ and one pair is defined for every variable $x,$ some predicate is satisfied for every $x.$ Thus $\phi$ is also true.  

Finally, the time of the reduction is equal to the size of $G.$ This is $2k+2=O(1)$ graphs, each of which is either $H$ or its complement. Hence $|G|=O(m+n +(n^2-m)+n)=O(n^2).$
\end{proof}

\hbdecidemchard*
\begin{proof}
We first reduce the instance of \fofee{} to \multiconn{} as in the proof of \lemref{fo-multiconn}. Let $G(V,E)$ be the multi-partite graph for \multiconn{} and $S,T$ the first and last parts of nodes of $G$. We add a sufficient number of nodes, referred as dummy nodes, to make G sparse. Let every node $x$ of $V\backslash T$ correspond to a distinct thread $t_x$ and form one write access event to a distinct variable $v_x$ in the thread. Let each node $t$ in $T$ also correspond to a write access event of the variable corresponding to the copy of $t$ in $S,$ and be in a new thread. Define $|E|$ locks, and for every edge $(a,b)\in E,$ let the events corresponding to $v_a$ and $v_b$ hold the lock $l_{(a,b)}$ corresponding to $(a,b).$ The trace $\tr$ for first lists all threads corresponding to the dummy nodes in some fixed arbitrary order, then the threads corresponding to nodes in $S,$ followed by those in each $Y_i,$ followed by those in each $Z_i,$ in a fixed arbitrary order, and finally those in $T$. 

This reduction is seen to be correct by observing that $G$ was modified to be the transitive reduction graph of $\tr,$ and the only $\HB$-race events can be the pairs of write events corresponding to the pairs of nodes given as input to \multiconn{}. Thus, each pair of events does not form an $\HB$-race if and only if $G$ has a path between its corresponding pair of nodes. 

To analyze the time of the reduction, first we see that the size of $\tr$ is the size of $G,$ with dummy nodes added to have $n=O(n^2),$ and hence $O(n^2).$ There are $O(n^2)$ variables, locks and threads in $\tr.$ If deciding if the given trace has an $\HB$-race has an $O((n^2)^{1+\epsilon})$ time algorithm, then \fofee{} can be solved in $O(n^{2+\epsilon'})$ time, which is $O(m^{1+\epsilon'})$ time for properties on dense structures.
\end{proof}
\section{Proofs of \cref{sec:syncp}}\label{sec:proofs_syncp}

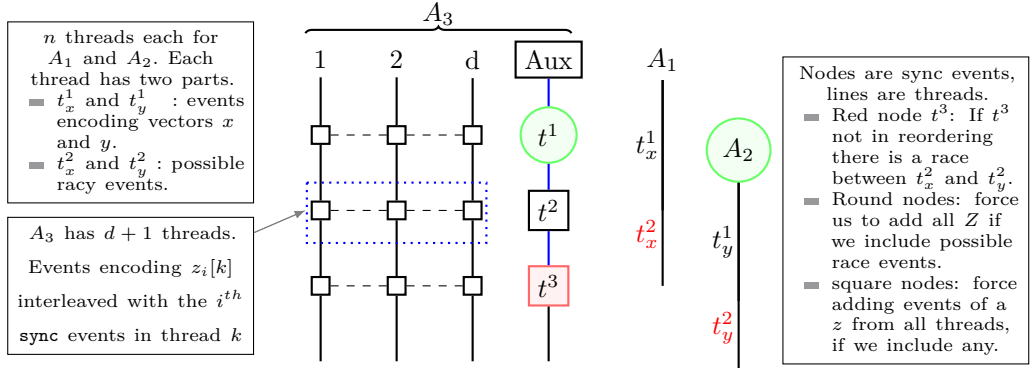
\begin{figure}
\begin{tikzpicture}[
roundnode/.style={circle, draw=green!60, fill=green!5, thick, minimum size=0.5mm},
squarednode/.style={rectangle, draw=black, thick, minimum size=2mm},
racenode/.style={rectangle, draw=red!60, fill=red!5, thick, minimum size=2mm},
textbox/.style = {rectangle, text centered, text width=3cm, draw=black},
tnode/.style = {rectangle}
]

\node [squarednode] (d1interz12) {};
\node [squarednode, right of=d1interz12] (d2interz12) {};
\node [squarednode, right of=d2interz12] (ddinterz12) {};

\node [tnode, above of=d1interz12] (init1) {1};
\node [tnode, above of=d2interz12] (init2) {2};
\node [tnode, above of=ddinterz12] (initd) {d};

\node [squarednode, below of=d1interz12] (d1interz23) {};
\node [squarednode, below of=d2interz12] (d2interz23) {};
\node [squarednode, below of=ddinterz12] (ddinterz23) {};
\node [squarednode, below of=d1interz23] (d1interzlast) {};
\node [squarednode, below of=d2interz23] (d2interzlast) {};
\node [squarednode, below of=ddinterz23] (ddinterzlast) {};

\draw [black, thick] (init1) -- (d1interz12);
\draw [black, thick] (d1interz12) -- (d1interz23);
\draw [black, thick] (d1interz23) -- (d1interzlast);
\draw [black, thick] (d1interzlast) -- +(0,-1);

\draw [black, thick] (init2) -- (d2interz12);
\draw [black, thick] (d2interz12) -- (d2interz23);
\draw [black, thick] (d2interz23) -- (d2interzlast);
\draw [black, thick] (d2interzlast) -- +(0,-1);

\draw [black, thick] (initd) -- (ddinterz12);
\draw [black, thick] (ddinterz12) -- (ddinterz23);
\draw [black, thick] (ddinterz23) -- (ddinterzlast);
\draw [black, thick] (ddinterzlast) -- +(0,-1);


\draw[blue,thick,dotted] ($(d1interz12.south west)+(-0.05,-0.5)$)  rectangle ($(ddinterz23.south east)+(0.05,-0.3)$);

\node [squarednode, right of=initd] (aux) {Aux};
\node [roundnode, below of=aux] (syncs) {$t^1$};
\node [squarednode, below of=syncs] (syncX) {$t^2$};
\node [racenode, below of=syncX] (syncrace) {$t^3$};
\draw[blue, thick] (aux) -- (syncs);
\draw[blue, thick] (syncs) -- (syncX);
\draw[blue, thick] (syncX) -- (syncrace);
\draw [black, thick] (syncrace) -- +(0,-1);

\node [textbox, below of=d1interz12,xshift=-2.5cm, yshift=-1cm] (zevent) {\scriptsize $A_3$ has $d+1$ threads. Events encoding $z_i[k]$ interleaved with the $i^{th}$ $\sync$ events in thread $k$};
\draw[arrow2] (zevent) -- ($(d1interz23)+(-0.2,0)$);

\draw [dashed] (d1interz12) -- (d2interz12);
\draw [dashed] (d2interz12) -- (ddinterz12);
\draw [dashed] (d1interz23) -- (d2interz23);
\draw [dashed] (d2interz23) -- (ddinterz23);
\draw [dashed] (d1interzlast) -- (d2interzlast);
\draw [dashed] (d2interzlast) -- (ddinterzlast);

\draw[thick, decoration={brace}, decorate] ($(init1)+(-0.2,0.35)$) -- ($(aux)+(0.3,0.35)$) node [midway, yshift=0.25cm] {$A_3$};

\node [textbox, above of=zevent, yshift=1.3cm] (XY) {\scriptsize $n$ threads each for $A_1$ and $A_2.$ Each thread has two parts.
\begin{itemize}
    \item $t_x^1$ and $t_y^1:$ events encoding vectors $x$ and $y.$ 
    \item $t_x^2$ and $t_y^2:$ possible racy events.
\end{itemize}};

\draw[arrow2] (zevent) -- ($(d1interz23)+(-0.2,0)$);

\node [tnode, right of=aux, xshift=0.5cm] (x1) {$A_1$};
\draw [black, thick] (x1) -- +(0,-2) node [midway, xshift=-0.2cm] {$t_x^1$};
\draw [black, thick] (x1) -- +(0,-3) node [near end, xshift=-0.2cm] {$\color{red} t_x^2$};

\node [roundnode, right of=x1, yshift=-1.2cm] (y1) {$A_2$};
\draw [black, thick] (y1) -- +(0,-2) node [midway, xshift=-0.2cm] {$t_y^1$};
\draw [black, thick] (y1) -- +(0,-3) node [near end, xshift=-0.2cm] {$\color{red} t_y^2$};

\node [textbox, right of=x1, xshift=2.2cm, yshift=-2cm] (labels) {\scriptsize Nodes are sync events, lines are threads. \begin{itemize}
    \item Red node $t^3$: If $t^3$ not in reordering there is a race between $t_x^2$ and $t_y^2$. 
\item Round nodes: force us to add all $Z$ if we include  possible race events.
\item square nodes: force adding events of a $z$ from all threads, if we include any.
\end{itemize}
};

\end{tikzpicture}
\caption{Intuition for the reduction from \orthvk{3} to $\sync$-preserving race detection. Thread \textit{Aux} allows forming a trace such that \orthvk{3} has a solution iff there is a $\sync$-preserving race.}
\figlabel{intuition-ov3-syncp}
\end{figure}

\syncpovthreehard*
\begin{proof}
Consider any sync-preserving correct reordering $\rho$ of $\tr$ that exposes a data race $(\wt(z), \rd(z))$ on the local traces $t_x$ and $t_y$.
The following statements are straightforward to verify based on the definition of sync-preserving correct reorderings.
\begin{compactenum}
\item\label{item:obs1}  For every $k\in[d]$, the first $\sync(\ell_k)$ event the trace $t_k$ is also in $\rho$.

\item\label{item:obs2} 
The auxiliary trace $t$ cannot have an open critical section in $\rho$.
This implies that for every trace $t_k$ with $k\in[d]$, the last event of $t_k$ in $\rho$ cannot be its $i^{th}$ $\sync(\ell_k)$ event, where $i$ is even.
Moreover, the number of $\sync(\ell_k)$ events in $\rho$ is the same for every trace $t_k$ with $k\in[d]$.
\end{compactenum}

First, consider that the \orthvk{3} instance has a solution, i.e., there exist $x\in A_1$, $y\in A_2$ and $z\in A_3$ such that $x,y,z$ are orthogonal, and we argue that $\tr$ has a data race that is also sync-preserving.
We construct a  sync-preserving correct reordering $\rho$ of $\tr$ that exposes the data race.
We only specify the local traces that exist in $\rho$, as their interleaving that constructs $\rho$ will be identical to the one in $\tr$ (in other words, we only specify the prefix up to which every local trace of $\tr$ is executed in $\rho$).
We execute the traces $t_x$ and $t_y$ all the way before the corresponding $\wt(z)$ and $\rd(z)$ events (hence we are exposing a race between these two events).
For every $k\in [d]$ if $z[k]=0$ or $x[k]=0$, we execute $t_k$ up to the $(2\cdot i)^{th}$ $\sync(\ell_k)$ event, where $i$ is such that $z$ is the $i^{th}$ vector of $A_3$.
On the other hand, if $y[k]=0$, we execute $t_k$ up to the first $\rel(\ell_k)$ event that appears after the $(2\cdot i)^{th}$ $\sync(\ell_k)$ event in $t^k$.
Finally, we execute $t^k$ until its $(i-1)^{th}$ $\rel(X)$ event.

It is easy to verify that $\rho$ is a valid correct reordering. 
Indeed, we have two open critical sections in the threads $t_x$ and $t_y$, on the locks $X$ and $Y$ respectively.
Moreover, for every $k\in[d]$, we have the following.
\begin{compactenum}
\item If $z[k]=0$, there are no other open critical sections.
\item If $z[k]=1$ and $x[k]=0$, there is one open critical section in the thread $z[k]$ on lock $l_k$.
\item If $z[k]=x[k]=1$ and $y[k]=0$, there is one open critical section in the thread $z[k]$ on lock $l'_k$.
\end{compactenum}

We now consider the opposite direction, i.e., assume that there is a sync-preserving race in $\tr$, and we argue that there exist $x\in A_1$, $y\in A_2$ and $z\in A_3$ such that $x,y,z$ are orthogonal.
Consider any sync-preserving correct reordering $\rho$ that exposes a race on the access events of two local traces $t_x$ and $t_y$.
Because of \cref{item:obs1} above, every trace $t_k$ is at least partially present in $\rho$.
Because of \cref{item:obs1} above, every such trace executes the same number of $\sync(\ell_k)$ events in $\rho$,
and this number is odd.
We argue that the triplet $x,y,z$ is orthogonal, where $z$ is the $i^{th}$ vector of $A_3$ such that each $t_k$ executes $2\cdot (i-1)+1$ $\sync(\ell_k)$ events in $\rho$.
Indeed, consider any $k\in[d]$ and assume that $x[k]=y[k]=1$.
If $z[k]=1$, then we have a $\acq(\ell)$ event in $t_k$ that immediately precedes its last $\sync(\ell_k)$ event.
Since $x[k]=1$, the trace $t_x$ also has an $\acq(l_k)$ event.
Since the $\acq(l_k)$ event of $t_x$ is after the $\acq(l_k)$ event of $t_k$ in $\tr$,
the matching $\rel(l_k)$ of $t_k$ must also be in $\rho$.
This implies that the $\acq(l'_k)$ event of $t_k$ that immediately succeeds its last $\rel(l_k)$ event is also in $\rho$.
Since $y[k]=1$, the trace $t_y$ also has an $\acq(l'_k)$ event.
Since the $\acq(l'_k)$ event of $t_y$ is after the $\acq(l'_k)$ event of $t_k$ in $\tr$,
the matching $\rel(l'_k)$ of $t_k$ must also be in $\rho$.
However, we now have another $\sync(\ell_k)$ event of $t_k$ in $\rho$, in particular, the $\sync(\ell_k)$ event that immediately precedes its last $\rel(l'_k)$ event.
But this results in an even number of $\sync(\ell_k)$ events of $t_x$ being present in $\rho$, which contradicts our observation in \cref{item:obs2}.
Thus, if $x[k]=y[k]=1$, we necessarily have that $z[k]=0$, and the triplet $x,y,z$ is orthogonal.

The desired result follows.
\end{proof}

\section{Proofs of \cref{sec:locking_discipline}}\label{sec:proofs_locking_discipline}

\lockcoverquadraticlowerbound*
\begin{proof}[Proof of \thmref{lock-cover-quadratic-lower-bound}]
To see why the reduction is correct, observe that if there a solution to \orthv{}, that is, a pair of vectors $x,y$ such that for all $k\in [d],$ $x[k]=0$ or $y[k]=0,$ implies that the corresponding events in $\tr,$ say $e_x$ and $e_y,$ have for each lock $k\in [d]$ either $e_x$ does not hold the lock or $e_y$ does not. As they have distinct thread ids too, $e_x$ and $e_y$ form two conflicting events with $\lheld{\tr}(e_1)\cap \lheld{\tr}(e_2)=\emptyset.$ Similarly, a lock-cover race implies that for every lock, one of the events in race do not hold the lock, hence have their corresponding coordinate in \orthv{} $0.$ The events are thus orthogonal to each other. 

Regarding the complexity, we have used $O(n\cdot d)$ time to construct a trace $\sigma$ with $\NumEvents=O(n\cdot d)$ events. 
If we can detect a lock-cover race in $\tr$ in $O(\NumEvents^{(2-\epsilon)})$ time, then \orthv{} can be solved in $O(n^{(2-\epsilon)}\cdot \poly(d))$ time, contradicting the \orthv{} hypothesis.
\end{proof}

\locksetsinglevariablelinear*
\begin{proof}
We first argue that the algorithm maintains the invariant stated in \cref{eq:lockset_invariant}.
The invariant for $A$ is trivial to verify.
Moreover, it is easy to see that, assuming that the invariant holds before processing an $\acq(\ell)$ or $\rel(\ell)$ event,
it also holds after processing that event.
Indeed, for an event $\acq(\ell)$, we have $\ell\in A$, and to maintain $C=\ov{A}\cap B$, we remove $\ell$ from $C$ if $\ell\in B$.
Similarly for an event $\rel(\ell)$.
To see that the invariant is maintained after processing an access event $\wt(x)/\rd(x)$, note that we have
\[
B\setminus C = B\cap \ov{C} = B\cap \left(\ov{B\cap \ov{A}}\right) = B\cap \left(\ov{B}\cup A\right) = B\cap A
\]
and thus updating $B\gets B\setminus C$ yields
\[
\locks{\tr}\cap \bigcap_{\substack{e' \in \accesses{\tr}(x) \\ e'\stricttrord{\tr} e}} \lheld{\tr}(e') \cap \lheld{tr}(e) = 
\locks{\tr}\cap \bigcap_{\substack{e' \in \accesses{\tr}(x) \\ e'\trord{\tr} e}} \lheld{\tr}(e')
\]
Finally, at this point we have $\ov{A}\cap B = \ov{A} \cap B\cap A=\emptyset$, thus the invariant also holds for $C$.

We now turn our attention to complexity.
Using a bit-set representation of the sets $A$, $B$ and $C$, it is clear that each of the operations except $\wt(x)/\rd(x)$ take constant time per event.
Each $\wt(x)/\rd(x)$ operation takes $O(|C|)$ time.
Note, however, that because of the previous invariant, every lock is removed from $B$ at most once, hence the total time for performing all set differences $B\gets B\setminus C$ is $O(\NumEvents + \NumLocks) = O(\NumEvents)$.
Thus the total time is $O(\NumEvents)$.
The desired result follows.
\end{proof}

\locksetquadraticlowerbound*
\begin{proof}
First, assume there is a solution to \hs{}, i.e., $\exists x_k\in X\ \forall y_i\in Y\ \exists j\in [d]\ x_k[j]=y_i[j]=1$. 
Then for the variable $z_k$, for every lock $\ell_i$, there is a thread $t_j$ that contains $\wt(z_k)$ (as $x_k[j]=1$) but does not contain lock $\ell_i$ (as $y_i[j]=1$).
Thus $\bigcap_{e \in \accesses{\tr}(z_k)} \lheld{\tr}(e) = \emptyset$, and we have a lock-set race on variable $z_k$
as there are at least two $\wt(z_k)$ conflicting events, one in the thread $t_j$ and the other in thread $t_0$.
For the opposite direction, assume that \hs{} does not have a solution, i.e., $\forall x_k\in X\ \exists y_i\in Y\ \forall j\in [d]\ (x_k[j]=0 \text{ or }y_i[j]=0)$. 
Then for each variable $z_k$, there is some lock $\ell_i$ such that every thread that contains a write event $\wt(z_k)$ (thus $x_k[j]=1$) also contains the lock $\ell_i$ (as necessarily $y_i[j]=0$). 
Hence, for every variable $z_k,$ some lock $\ell_i$ is held by all its access events. 
Thus $\tr$ does not have a lock-set race. 

Regarding the complexity, we have created a trace $\tr$ with $\NumEvents = O(n\cdot d)$ events in $O(n\cdot d)$ time.
Thus, any $O(\NumEvents^{(2-\epsilon)})$ time algorithm for \hs{} implies an $O(n^{(2-\epsilon)}\cdot \poly(d))$ time algorithm for \hs{}, contradicting the \hs{} hypothesis.
\end{proof}

\end{document}